\pdfoutput=1
\documentclass[11pt]{article}
\usepackage{xspace}
\DeclareMathAlphabet{\mathbbold}{U}{bbold}{m}{n}
\usepackage{amssymb,amsmath,amsthm,graphicx,bbm}
\usepackage{mathtools}
\usepackage[usenames,dvipsnames,svgnames,table]{xcolor}
\usepackage{thmtools, thm-restate}
\usepackage[margin=1.in]{geometry}
\usepackage{multirow,array}
\usepackage{colortbl}
\allowdisplaybreaks
\usepackage{microtype}
\usepackage[pagebackref]{hyperref}
\hypersetup{
    pdftitle={}, 
    pdfauthor={}, 
    colorlinks=true, 
    linkcolor=blue, 
    citecolor=blue, 
    filecolor=blue, 
    urlcolor=blue 
}

\usepackage{caption}

\renewcommand{\backref}[1]{}

\renewcommand{\backrefalt}[4]{%
\ifcase #1 %
\or
[p.\ #2]%
\else
[pp.\ #2]%
\fi}

\usepackage[capitalise,nameinlink]{cleveref}

\let\oldproofname=\proofname
\renewcommand{\proofname}{\rm\bf{\oldproofname}}

\makeatletter
\newcommand{\para}{%
 \@startsection{paragraph}{4}%
 {\z@}{2ex \@plus 3.3ex \@minus .2ex}{-1em}%
 {\normalfont\normalsize\bfseries}%
}
\makeatother

\usepackage{algorithm}
\usepackage[noend]{algpseudocode}
\usepackage{tikz,tikz-qtree}
\usepackage{wrapfig}
\usepackage{enumitem}

\theoremstyle{plain}
\newtheorem{theorem}{Theorem}

\newtheorem{proposition}[theorem]{Proposition}
\newtheorem{lemma}[theorem]{Lemma}

\newtheorem{conjecture}{Conjecture}

\newtheorem{definition}[theorem]{Definition}

\theoremstyle{remark}

\theoremstyle{plain}

\def\Reals{{\mathbb{R}}} 
\def\N{{\mathbb{N}}} 

\renewcommand{\Pr}{\mathop{\bf Pr\/}}
\newcommand{\E}{\mathop{\bf E\/}}

\newcommand{\id}{\mathbbold{1}}

\newcommand{\OR}{\textsc{or}}
\newcommand{\AND}{\textsc{and}}
\newcommand{\Parity}{\textsc{parity}}

\newcommand{\eps}{\epsilon}

\newcommand{\transpose}{{\mathsf{T}}}

\renewcommand{\hat}{\widehat}
\renewcommand{\tilde}{\widetilde}

\newcommand{\B}{\{0,1\}}

\let\OldLambda\lambda
\let\lambda\relax
\DeclareMathOperator{\lambda}{\OldLambda}
\DeclareMathOperator{\D}{\mathsf{D}}
\DeclareMathOperator{\C}{\mathsf{C}}
\DeclareMathOperator{\R}{\mathsf{R}}
\DeclareMathOperator{\Q}{\mathsf{Q}}
\DeclareMathOperator{\UC}{\mathsf{UC}}
\DeclareMathOperator{\RC}{\mathsf{RC}}

\DeclareMathOperator{\s}{\mathsf{s}}
\DeclareMathOperator{\bs}{\mathsf{bs}}
\DeclareMathOperator{\adeg}{\mathsf{\widetilde{deg}}}
\let\deg\relax
\DeclareMathOperator{\deg}{\mathsf{deg}}
\DeclareMathOperator{\Adv}{\mathsf{Adv}}
\DeclareMathOperator{\SA}{\mathsf{SA}}
\DeclareMathOperator{\MM}{\mathsf{MM}}
\DeclareMathOperator{\SWA}{\mathsf{SWA}}
\DeclareMathOperator{\GSA}{\mathsf{GSA}}
\DeclareMathOperator{\K}{\mathsf{K}}
\DeclareMathOperator{\Dom}{Dom}
\DeclareMathOperator{\tr}{tr}
\DeclareMathOperator{\diag}{diag}

\newcommand{\tR}{\tilde{R}}

\newcommand{\norm}[1]{\lVert#1\rVert}

\definecolor{conj}{HTML}{C2C0BF} 
\definecolor{open}{HTML}{A31F34} 

\newcommand{\ct}[2]{%
\cellcolor{white}%
\begin{tabular}[t]{@{}c@{}}#1\\[-5pt]#2\end{tabular}%
}

\newcommand{\co}[2]{%
\cellcolor{conj!70}%
\begin{tabular}[t]{@{}c@{}}#1\\[-5pt]#2\end{tabular}%
}

\newcommand{\cc}[2]{%
\cellcolor{conj!70}%
\begin{tabular}[t]{@{}c@{}}#1\\[-5pt]#2\end{tabular}%
}

\newcommand{\smcite}[1]{{\scriptsize \cite{#1}}}
\newcommand{\newbound}[1]{\colorbox{open!50}{#1}}
\newcommand{\Huangbound}[1]{\colorbox{green!50}{#1}}

\title{Degree vs.~Approximate Degree and \\ Quantum Implications of Huang's Sensitivity Theorem%
\footnote{This subsumes an earlier preprint by a subset of the authors \cite{ABKT20}.}
}
\author{%
\hspace{5em} Scott Aaronson\footnote{Department of Computer Science, University of Texas at Austin. \texttt{aaronson@cs.utexas.edu}} 
\and 
Shalev Ben-David\footnote{University of Waterloo.
\texttt{shalev.b@uwaterloo.ca}} 
\and 
Robin Kothari\footnote{Microsoft Quantum and Microsoft Research. \texttt{robin.kothari@microsoft.com}}
\hspace{5em}
\and 
Shravas Rao\footnote{Northwestern University.  \texttt{shravas@northwestern.edu}}
\and
Avishay Tal\footnote{%
Department of Electrical Engineering and Computer Sciences, University of California at Berkeley. \texttt{atal@berkeley.edu}}
}
\date{\vspace{-3ex}}

\begin{document}
\maketitle
\begin{abstract}
Based on the recent breakthrough of Huang (2019), we show that for any total Boolean function $f$, 
\begin{enumerate}
    \item $\deg(f) = O(\adeg(f)^2)$: The degree of $f$ is at most quadratic in the approximate degree of $f$. This is optimal as witnessed by the OR function.
    \item $\D(f) = O(\Q(f)^4)$: The deterministic query complexity of $f$ is at most quartic in the quantum query complexity of $f$. This matches the known separation (up to log factors) due to Ambainis, Balodis, Belovs, Lee, Santha, and Smotrovs (2017).
\end{enumerate}
We apply these results to resolve the quantum analogue of the Aanderaa--Karp--Rosenberg conjecture. We show that if $f$ is a nontrivial monotone graph property of an $n$-vertex graph specified by its adjacency matrix, then $\Q(f)=\Omega(n)$, which is also optimal. We also show that the approximate degree of any read-once formula on $n$ variables is $\Theta(\sqrt{n})$.
\end{abstract}

\section{Introduction}
Last year, Huang resolved a major open problem in the analysis of Boolean functions called the \emph{sensitivity conjecture}~\cite{Huang2019}, which was open for nearly 30 years \cite{NS94}. Surprisingly, Huang's elegant proof takes less than 2 pages---truly a ``proof from the book.'' 
Specifically, Huang showed that for any total Boolean function, which is a function $f:\B^n\to\B$, we have
\begin{equation}\label{eq:Huang}
    \deg(f) \leq \s(f)^2,
\end{equation}
where $\deg(f)$ is the real degree of $f$ and $\s(f)$ is the (maximum) sensitivity of $f$. These measures and other standard measures appearing in this introduction are defined in \Cref{sec:prelim}.

In this paper, we describe some implications of Huang's resolution of the sensitivity conjecture to polynomial-based complexity measures of Boolean functions and quantum query complexity. 
Our first observation is that Huang actually proves a stronger claim than \cref{eq:Huang}, in which $\s(f)$ can be replaced by $\lambda(f)$, a spectral relaxation of sensitivity we define in \Cref{def:sengraph}. 

\begin{restatable}{theorem}{Huang}\label{lemma:Huang}\label{thm:Huang}
For all  Boolean functions $f:\B^n \to \B$, we have
	$\deg(f) \le \lambda(f)^2$.
\end{restatable}

In short, while $s(f)$ can be viewed as the maximum number of $1$s in any row or column of a certain Boolean matrix, $\lambda(f)$ is the largest eigenvalue of that matrix, which could potentially be smaller.
This observation has several implications because, as we show, $\lambda(f)$ lower bounds many other complexity measures. 
One of the messages of this work is that $\lambda(f)$ is an interesting complexity measure and can be used to establish relationships between other complexity measures.

We use this observation to prove two main results: 
Our first result is an optimal relationship between deterministic and quantum query complexity for total functions, and our second result is an optimal relationship between degree and approximate degree for total functions. 
We then apply the first result to prove the quantum analogue of the Aanderaa--Karp--Rosenberg conjecture and apply the second result to show that the approximate degree of any read-once formula is $\Theta(\sqrt{n})$.

\para{Deterministic vs.\ quantum query complexity.}
We know from the seminal results of Nisan \cite{Nisan91}, Nisan and Szegedy~\cite{NS94}, and Beals et al.\ \cite{BBCMW01} that for any total Boolean function $f$,
the deterministic query complexity, $\D(f)$, and quantum query complexity, $\Q(f)$, satisfy%
\footnote{%
This means that for total functions, quantum query algorithms can only outperform classical query algorithms by a polynomial factor. On the other hand, for partial functions, which are defined on a subset of $\B^n$, exponential and even larger speedups are possible.}
\begin{equation}\label{eq:Beals}
    \D(f)=O(\Q(f)^6).
\end{equation}
Grover's algorithm \cite{Grover96} shows that for the OR function, a quadratic separation between $\D(f)$ and $\Q(f)$ is possible. This was the best known quantum speedup for total functions until  Ambainis et al.\ \cite{ABB+15} constructed a total function $f$ with 
\begin{equation}\label{eq:ABB}
\D(f) = \tilde{\Omega}(\Q(f)^{4}).    
\end{equation}

We show that the quartic separation (up to log factors) in \cref{eq:ABB} is actually the best possible.
\begin{theorem}\label{thm:D vs. Q}
For all Boolean functions $f:\B^n \to \B$, we have 
$\D(f) = O(\Q(f)^4)$. 	
\end{theorem}

We deduce \Cref{thm:D vs. Q} as a corollary of a new tight 
relationship between $\deg(f)$ and $\Q(f)$:

\begin{theorem}\label{thm:deg vs. Q}
For all Boolean functions $f:\B^n \to \B$, we have 
$\deg(f) = O(\Q(f)^2)$. 	
\end{theorem}

Observe that \Cref{thm:deg vs. Q} is tight for the OR function on $n$ variables, whose degree is $n$ and whose quantum query complexity is $\Theta(\sqrt{n})$ \cite{Grover96,BBBV97}. Prior to this work, the best relation between $\deg(f)$ and $\Q(f)$ was a sixth power relation, $\deg(f) = O(\Q(f)^6)$, which follows from \cref{eq:Beals}.

As discussed, our proof relies on the restatement of Huang's result (\Cref{thm:Huang}), showing that $\deg(f)\leq \lambda(f)^2$. We show (in \Cref{lem:lambdaQ}) that the measure $\lambda(f)$ lower bounds the original quantum adversary method of Ambainis~\cite{Amb02}, which in turn lower bounds $\Q(f)$.

We now show how \Cref{thm:D vs. Q} straightforwardly follows from \Cref{thm:deg vs. Q} using two previously known connections between complexity measures of Boolean functions.

\begin{proof}[Proof of \Cref{thm:D vs. Q} assuming \Cref{thm:deg vs. Q}]
Midrijanis \cite{Mid04} showed that for all total functions $f$, 
\begin{equation}\label{eq:Mid}
\D(f) \leq \bs(f) \deg(f),
\end{equation} 
where $\bs(f)$ is the block sensitivity of $f$.
	
\Cref{thm:deg vs. Q} shows that $\deg(f) = O(\Q(f)^2)$. Combining the relationship between block sensitivity and approximate degree from \cite{NS94} with the results of \cite{BBCMW01}, we get that $\bs(f) =  O(\Q(f)^2)$. (This can also be proved directly using the lower bound method in \cite{BBBV97}.)

Combining these three inequalities yields $\D(f) = O(\Q(f)^4)$ for all total Boolean functions $f$.
\end{proof}

We establish \Cref{thm:deg vs. Q} in \Cref{sec:mainproof} using \Cref{thm:Huang} and the spectral adversary method in quantum query complexity \cite{BSS03}.

\para{Degree vs. approximate degree.}
We also know from the works of Nisan and Szegedy~\cite{NS94} and Beals et al.~\cite{BBCMW01}, that for any total Boolean function $f$, 
\begin{equation}
    \deg(f) = O(\adeg(f)^6),
\end{equation}
where $\deg(f)$ and $\adeg(f)$ are the exact degree and approximate degree of $f$ respectively (defined in \Cref{sec:prelim}).
We show that this relationship can be also significantly improved.

\begin{theorem}\label{thm:degadeg}
For all Boolean functions $f:\B^n \to \B$, we have $\deg(f) = O(\adeg(f)^2)$.
\end{theorem}

This relationship is optimal since it is saturated by the OR function on $n$ bits that has degree $n$ and approximate degree $\Theta(\sqrt{n})$~\cite{NS94}.

\Cref{thm:degadeg} follows by combining $\deg(f) \leq \lambda(f)^2$ (\Cref{thm:Huang}) with $\lambda(f) = O(\adeg(f))$, which we prove in \Cref{sec:degree}. 
This is the most technically challenging part of this paper, and we provide two proofs of this claim.
The first proof (\Cref{thm:lambda lower bounds adeg}) is arguably simpler, but it is not self contained and uses Sherstov's composition theorem for approximate degree~\cite{She13}, and has a large constant hidden in the big Oh. 
The second proof (\Cref{thm:tightadeg}) does not rely on this result and achieves the optimal constant inside the big Oh.

Observe that because approximate degree lower bounds quantum query complexity, \Cref{thm:degadeg} also implies \Cref{thm:deg vs. Q} (and hence \Cref{thm:D vs. Q}). Although \Cref{thm:deg vs. Q} is a consequence of this, our proof of \Cref{thm:deg vs. Q} is much simpler and additionally proves that $\lambda(f)$ lower bounds the positive-weights adversary method~\cite{Amb02}, which is not implied by the proof of \Cref{thm:degadeg}.

\para{Applications.}
In \Cref{sec:AKR}, we  use \Cref{thm:deg vs. Q} to prove the quantum analogue of the famous \textit{Aanderaa--Karp--Rosenberg conjecture}. Briefly, this conjecture is about the minimum possible query complexity of a nontrivial monotone graph property, for graphs specified by their adjacency matrices.

There are variants of the conjecture for different models of computation. For example, the randomized variant of the Aanderaa--Karp--Rosenberg conjecture, attributed to Karp~\cite[Conjecture 1.2]{SW86} and Yao~\cite[Remark (2)]{Yao77}, states that for all nontrivial monotone graph properties $f$, we have $\R(f) = \Omega(n^{2})$. Following a long line of work, the current best lower bound is $\R(f) = \Omega(n^{4/3} \log^{1/3}n)$ due to Chakrabarti and Khot~\cite{CK01}.

The quantum version of the conjecture was raised by Buhrman, Cleve, de Wolf, and Zalka~\cite{BCdWZ99}, who observed that the best we could hope for is $\Q(f) = \Omega(n)$, because the nontrivial monotone graph property ``contains at least one edge'' can be decided with $O(n)$ queries using Grover's algorithm.
Buhrman et al.~\cite{BCdWZ99} also showed that all nontrivial monotone graph properties satisfy $\Q(f)=\Omega(\sqrt{n})$. The current best bound is $\Q(f) = \Omega(n^{2/3}\log^{1/6}n)$, which is credited to Yao in \cite{MSS07}.
We resolve this conjecture by showing an optimal $\Omega(n)$ lower bound.

\begin{theorem}\label{thm:qAKR}
Let $f:\B^{\binom{n}{2}} \to \B$ be a nontrivial monotone graph property. Then $\Q(f)=\Omega(n)$.
\end{theorem}

\Cref{thm:qAKR} follows by combining \Cref{thm:deg vs. Q} with a known quadratic lower bound on the degree of monotone graph properties.

In \Cref{sec:readonce}, we use \Cref{thm:degadeg} to completely characterize the approximate degree of any read-once formula. It is known that the quantum query complexity of any read-once formula on $n$ variables is $\Theta(\sqrt{n})$~\cite{BS04,Rei11}. 
It has long been conjectured that the approximate degree of any read-once formula is also $\Theta(\sqrt{n})$. It has taken much effort to establish this even for special read-once formulas.
For example, the conjecture was proved for the simple depth-two read-once formula $\AND \circ \OR$ in 2013~\cite{BT13,She13a}. 
This result was later extended to all constant-depth balanced read-once formulas~\cite{BT15} and then to constant-depth unbalanced read-once formulas~\cite{BBGK18}. We resolve this question for all read-once formulas.

\begin{restatable}{theorem}{readonce}\label{thm:readonce}
For any read-once formula $f:\B^n \to \B$, we have $\adeg(f) = \Theta(\sqrt{n})$.
\end{restatable}

\subsection{Known relations and separations}

\setlength{\intextsep}{0pt}%
\setlength{\columnsep}{10pt}%
\begin{wrapfigure}{r}{0.35\textwidth}
\centering
\vspace{-5ex}
  \begin{tikzpicture}[x=1cm,y=1cm]

     \node (D) at(2,5){$\D$};
     \node (Rz) at(1,4){$\R_0$};
     \node (QE) at(3,4){$\Q_E$};
     \node (C) at(0,3){$\C$};
     \node (R) at(2,3){$\R$};
     \node (RC) at(1,2){$\RC$};
     \node (bs) at(1,1){$\bs$};
     \node (s) at(1,0){$\s$};
     \node (l) at(2,-1){$\lambda$};
     \node (deg) at(4,3){$\deg$};
     \node (Q) at(3,2){$\Q$};
     \node (adeg) at(4,0.9){$\adeg$};

     \path[-] (Rz) edge (D);
     \path[-] (QE) edge  (D);
     \path[-] (C) edge (Rz);
     \path[-] (R) edge (Rz);
     \path[-] (RC) edge (C);
     \path[-] (RC) edge (R);
     \path[-] (bs) edge (RC);
     \path[-] (bs) edge (s);
     \path[-] (s) edge (l);
     \path[-] (adeg) edge (l);
     \path[-] (deg) edge (QE);
     \path[-] (Q) edge (QE);
     \path[-] (Q) edge (R);
     \path[-] (adeg) edge (Q);
     \path[-] (adeg) edge (deg);
  \end{tikzpicture}
    \caption{Relations between complexity measures. An upward line from a measure $M_1(f)$ to $M_2(f)$ denotes $M_1(f) = O(M_2(f))$ for all total functions $f$.\label{fig:rel}}
\end{wrapfigure}
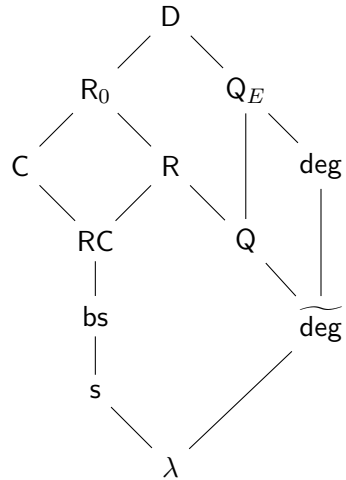

\Cref{tab:sep} summarizes the known relations and separations between complexity measures studied in this paper (and more). This is an update to a similar table that appears in \cite{ABK16} with the addition of $\s(f)$ and $\lambda(f)$. 
Definitions and additional details about interpreting the table can be found in \cite{ABK16}.

For all the separations claimed in the table, we provide an example of a separating function or a citation to construction of such a function. 
All the relationships in the table follow by combining the relationships depicted in \Cref{fig:rel} and the following inequalities that hold for all total Boolean functions:

\begin{itemize}[noitemsep,topsep=4pt]
\item $\C(f) \leq \bs(f)\s(f)$ \cite{Nisan91}
\item $\D(f) \leq \bs(f)\C(f)$ \cite{BBCMW01}
\item $\D(f)\leq \bs(f)\deg(f)$ \cite{Mid04}
\item $\RC(f) = O(\adeg(f)^2)$ \cite{KT16}
\item $\R_0(f) = O(\R(f) \s(f) \log \RC(f))$ \cite{KT16}
\item $\deg(f) \leq \lambda(f)^2$ \cite{Huang2019}
\item $\s(f) \leq \lambda(f)^2$ (\Cref{lem:lambdalower})
\end{itemize}

\setlength{\tabcolsep}{2pt}
\renewcommand{\arraystretch}{1.3}
\begin{table}
\hspace{-1em}
\begin{minipage}{\linewidth}
\begin{center}
\caption{Best known separations between complexity measures}
\begin{tabular}{r||c|c|c|c|c|c|c|c|c|c|c|c}
{} & $\D$ & $\R_0$ & $\R$ & $\C$ & $\RC$ & $\bs$ & $\s$ & $\lambda$ & $\Q_E$ & $\deg$ & $\Q$ & $\adeg$ \\ \hline\hline

$\D$ &
\cellcolor{darkgray} & 
\ct{2, 2}{\smcite{ABB+15}} & 
\cc{2, 3}{\smcite{ABB+15}} &
\ct{2, 2}{$\wedge\circ\vee$} & 
\cc{2, 3}{$\wedge\circ\vee$} & 
\cc{2, 3}{$\wedge\circ\vee$} &
\co{3,\Huangbound{6}}{\smcite{BHT17}} &
\co{4,\Huangbound{6}}{\smcite{ABB+15}} &
\co{2, 3}{\smcite{ABB+15}} & 
\co{2, 3}{\smcite{GPW15}} & 
\ct{4,\newbound{4}}{\smcite{ABB+15}} &
\ct{4,\newbound{4}}{\smcite{ABB+15}} \\ \hline

$\R_0$ &
\ct{1, 1}{$\oplus$} &
\cellcolor{darkgray} & 
\ct{2, 2}{\smcite{ABB+15}} &
\ct{2, 2}{$\wedge\circ\vee$} & 
\cc{2, 3}{$\wedge\circ\vee$} & 
\cc{2, 3}{$\wedge\circ\vee$} &
\co{3,\Huangbound{6}}{\smcite{BHT17}} &
\co{4,\Huangbound{6}}{\smcite{ABB+15}} &
\co{2, 3}{\smcite{ABB+15}} &
\co{2, 3}{\smcite{GJPW15}} &
\co{3,\newbound{4}}{\smcite{ABB+15}} & 
\ct{4,\newbound{4}}{\smcite{ABB+15}} \\ \hline

$\R$ &
\ct{1, 1}{$\oplus$} & 
\ct{1, 1}{$\oplus$} & 
\cellcolor{darkgray} & 
\ct{2, 2}{$\wedge\circ\vee$} &
\cc{2, 3}{$\wedge\circ\vee$} & 
\cc{2, 3}{$\wedge\circ\vee$} &
\co{3,\Huangbound{6}}{\smcite{BHT17}} &
\co{4,\Huangbound{6}}{\smcite{ABB+15}} &
\co{$\frac{3}{2}$, 3}{\smcite{ABB+15}} &
\co{2, 3}{\smcite{GJPW15}} &
\co{$3$,\newbound{4}\vspace{-0.75ex}}{\smcite{BS20}\\[-2ex]\smcite{SSW20}\vspace{-0.4ex}} & 
\ct{4,\newbound{4}}{\smcite{ABB+15}} \\ \hline

$\C$ &
\ct{1, 1}{$\oplus$} & 
\ct{1, 1}{$\oplus$} & 
\co{1, 2}{$\oplus$} & 
\cellcolor{darkgray} &
\ct{2, 2}{\smcite{GSS13}} & 
\ct{2, 2}{\smcite{GSS13}} &
\co{2.22,\Huangbound{5}}{\smcite{BHT17}} &
\co{2.44,\Huangbound{6}}{\smcite{BHT17}\footnote{%
\cite{BHT17} exhibited a family of functions
satisfying $\C(f)=\Omega(\UC_{\min}(f)^{1.22})$,
and they also gave a transformation which modifies
a function $f$ in a way that causes $\s_1(f)$
to become $O(1)$, causes $\s_0(f)$ to become
at most $\UC_{\min}(f)$, and does not decrease
$\C(f)$. Since we have
$\lambda(f)\le\sqrt{\s_0(f)\s_1(f)}$ by
\Cref{lem:lambda_s},
this implies a power $2.44$ separation between
$\lambda(f)$ and $\C(f)$.
}} &
\co{1.15, 3}{\smcite{Amb13}}  & 
\co{1.63, 3}{\smcite{NW95}} & 
\co{2, 4}{$\wedge$} & 
\co{2, 4}{$\wedge$} \\ \hline

$\RC$ &
\ct{1, 1}{$\oplus$} & 
\ct{1, 1}{$\oplus$} & 
\ct{1, 1}{$\oplus$} & 
\ct{1, 1}{$\oplus$} & 
\cellcolor{darkgray} & 
\co{$\frac{3}{2}$, 2}{\smcite{GSS13}} &
\co{2,\Huangbound{4}}{\smcite{Rub95}} &
\co{2,\Huangbound{4}}{$\wedge$} &
\co{1.15, 2}{\smcite{Amb13}} & 
\co{1.63, 2}{\smcite{NW95}} & 
\ct{2, 2}{$\wedge$} & 
\ct{2, 2}{$\wedge$} \\ \hline

$\bs$ &
\ct{1, 1}{$\oplus$} & 
\ct{1, 1}{$\oplus$} & 
\ct{1, 1}{$\oplus$} & 
\ct{1, 1}{$\oplus$} & 
\ct{1, 1}{$\oplus$} & 
\cellcolor{darkgray} &
\co{2,\Huangbound{4}}{\smcite{Rub95}} &
\co{2,\Huangbound{4}}{$\wedge$} &
\co{1.15, 2}{\smcite{Amb13}} & 
\co{1.63, 2}{\smcite{NW95}} & 
\ct{2, 2}{$\wedge$} & 
\ct{2, 2}{$\wedge$} \\ \hline

$\s$ &
\ct{1, 1}{$\oplus$} &
\ct{1, 1}{$\oplus$} &
\ct{1, 1}{$\oplus$} &
\ct{1, 1}{$\oplus$} &
\ct{1, 1}{$\oplus$} &
\ct{1, 1}{$\oplus$} &
\cellcolor{darkgray} &
\ct{2, 2}{$\wedge$} &
\co{1.15, 2}{\smcite{Amb13}} &
\co{1.63, 2}{\smcite{NW95}}  &
\ct{2, 2}{$\wedge$} & 
\ct{2, 2}{$\wedge$} \\ \hline

$\lambda$ &
\ct{1, 1}{$\oplus$} &
\ct{1, 1}{$\oplus$} &
\ct{1, 1}{$\oplus$} &
\ct{1, 1}{$\oplus$} &
\ct{1, 1}{$\oplus$} &
\ct{1, 1}{$\oplus$} &
\ct{1, 1}{$\oplus$} &
\cellcolor{darkgray} &
\ct{1,\newbound{1}}{$\oplus$} &
\ct{1,\newbound{1}}{$\oplus$} &
\ct{1,\newbound{1}}{$\oplus$} &
\ct{1,\newbound{1}}{$\oplus$}  \\ \hline

$\Q_E$ &
\ct{1, 1}{$\oplus$} &
\co{1.33, 2}{$\bar{\wedge}$-tree} &
\co{1.33, 3}{$\bar{\wedge}$-tree} &
\ct{2, 2}{$\wedge\circ\vee$} &
\cc{2, 3}{$\wedge\circ\vee$} &
\cc{2, 3}{$\wedge\circ\vee$} &
\co{3,\Huangbound{6}}{\smcite{BHT17}} &
\co{4,\Huangbound{6}}{\smcite{ABK16}} &
\cellcolor{darkgray} &
\co{2, 3}{\smcite{ABK16}} &
\co{2,\newbound{4}}{$\wedge$} &
\ct{4,\newbound{4}}{\smcite{ABK16}} \\ \hline

$\deg$ & 
\ct{1, 1}{$\oplus$} &
\co{1.33, 2}{$\bar{\wedge}$-tree} &
\co{1.33,\Huangbound{2}}{$\bar{\wedge}$-tree} &
\ct{2, 2}{$\wedge\circ\vee$} &
\ct{2,\Huangbound{2}}{$\wedge\circ\vee$} &
\ct{2,\Huangbound{2}}{$\wedge\circ\vee$} &
\ct{2,\Huangbound{2}}{$\wedge\circ\vee$} &
\ct{2,\Huangbound{2}}{$\wedge$} &
\ct{1, 1}{$\oplus$}&
\cellcolor{darkgray} &
\ct{2,\newbound{2}}{$\wedge$} &
\ct{2,\newbound{2}}{$\wedge$} \\ \hline

$\Q$ &
\ct{1, 1}{$\oplus$} &
\ct{1, 1}{$\oplus$} &
\ct{1, 1}{$\oplus$} &
\ct{2, 2}{\smcite{ABK16}} &
\cc{2, 3}{\smcite{ABK16}} &
\cc{2, 3}{\smcite{ABK16}} &
\co{3,\Huangbound{6}}{\smcite{BHT17}} &
\co{4,\Huangbound{6}}{\smcite{ABK16}} &
\ct{1, 1}{$\oplus$} &
\co{2, 3}{\smcite{ABK16}} &
\cellcolor{darkgray} &
\ct{4,\newbound{4}}{\smcite{ABK16}} \\ \hline

\raisebox{-2pt}{$\adeg$} &
\ct{1, 1}{$\oplus$} &
\ct{1, 1}{$\oplus$} &
\ct{1, 1}{$\oplus$} &
\ct{$2$, 2}{\smcite{BT17}} &
\ct{$2$,\Huangbound{2}}{\smcite{BT17}} &
\ct{$2$,\Huangbound{2}}{\smcite{BT17}} &
\ct{2,\Huangbound{2}}{\smcite{BT17}} &
\ct{2,\Huangbound{2}}{\smcite{BT17}} &
\ct{1, 1}{$\oplus$} &
\ct{1, 1}{$\oplus$}  &
\ct{1, 1}{$\oplus$} &
\cellcolor{darkgray} \\ 
\end{tabular}
\label{tab:sep}
\end{center}
\begin{itemize}[itemsep=3pt]
    \item  An entry $a,b$ in the row $M_1$ and column $M_2$ roughly means that there exists a function $g$ with $M_1(g) \geq M_2(g)^{a-o(1)}$, and for all total functions $f$, $M_1(f) \leq M_2(f)^{b+o(1)}$ (see \cite{ABK16} for a precise definition). For example, the $3,4$ entry at row $\R$ and column $\Q$ means that the maximum possible separation between $\R$ and $\Q$ is at least cubic and at most quartic.
    \item {The second row of each cell contains an example of a function that achieves the separation (or a citation to an example), 
    where $\oplus = \textsc{parity}$, 
    $\wedge = \textsc{and}$, 
    $\vee = \textsc{or}$, 
    $\wedge \circ \vee = \textsc{and-or}$, 
    and $\bar{\wedge}$-tree is the balanced \textsc{nand}-tree function.}
    \item Cells have a white background if the relationship is optimal and a gray background otherwise.
    \item Entries with a \Huangbound{green} background follow from Huang's result. Entries with a \newbound{red} background follow from this work.
\end{itemize}
\end{minipage}
\end{table}
\renewcommand{\arraystretch}{1}

\subsection{Paper organization}

\Cref{sec:prelim} contains preliminaries and definitions of various complexity measures that appear in this paper. \Cref{sec:mainproof} reproves \Cref{thm:Huang}, which follows Huang's original proof~\cite{Huang2019}, and then shows that $\lambda(f) = O(Q(f))$, which establishes \Cref{thm:deg vs. Q}. \Cref{sec:degree} establishes $\lambda(f) =O(\adeg(f))$, which implies \Cref{thm:degadeg}. \Cref{sec:AKR}  gives some background and motivation for the Aanderaa--Karp--Rosenberg conjecture and proves \Cref{thm:qAKR}. \Cref{sec:readonce} establishes \Cref{thm:readonce}. We end with some open problems in \Cref{sec:open}. \Cref{sec:properties} describes some properties of $\lambda(f)$, its many equivalent formulations, and its relationship with other complexity measures.

\section{Preliminaries}\label{sec:prelim}

\subsection{Query complexity}

Let $f:\B^n \to\B$ be a Boolean function.
Let $A$ be a deterministic
algorithm that computes $f(x)$ on input
$x\in\B^n$ by making queries to the bits of $x$.
The worst-case number of queries $A$ makes (over choices
of $x$) is the query complexity of $A$. The minimum query
complexity of any deterministic algorithm computing $f$
is the deterministic query complexity of $f$, denoted
by $\D(f)$.

We define the bounded-error
randomized (respectively quantum) query complexity
of $f$, denoted by $\R(f)$ (respectively $\Q(f)$),
in an analogous way. We say an algorithm $A$ computes $f$
with bounded error if $\Pr[A(x)=f(x)]\geq 2/3$ for all
$x\in\B^n$, where the probability is over the internal
randomness of $A$. Then $\R(f)$ (respectively $\Q(f)$)
is the minimum number of queries required by any
randomized (respectively quantum) algorithm that computes $f$
with bounded error. It is clear that $\Q(f)\leq \R(f)\leq \D(f)$.
For more details on these measures,
see the survey by Buhrman and de Wolf \cite{BdW02}.

\subsection{Sensitivity and block sensitivity}

Let $f:\B^n\to\B$ be a Boolean function, and let $x\in\B^n$
be a string. A block is a subset of $[n]$. We say that a block $B\subseteq [n]$
is sensitive for $x$ (with respect to $f$) if
$f(x \oplus \mathbbold{1}_B)\neq f(x)$, where $\mathbbold{1}_B$ is the $n$-bit string that
is $1$ on bits in $B$ and $0$ otherwise. 
We say a bit $i$ is sensitive for $x$
if the block $\{i\}$ is sensitive for $x$.
The maximum number of disjoint blocks
that are all sensitive for $x$
is called the block sensitivity of $x$ (with respect to $f$),
denoted by $\bs_x(f)$. The number of sensitive bits for $x$
is called the sensitivity of $x$,
denoted by $\s_x(f)$. Clearly, $\bs_x(f)\geq\s_x(f)$,
since $\s_x(f)$ is has the same definition as $\bs_x(f)$
except that the size of the blocks is restricted to $1$.
We define $\s(f) = \max_{x\in \B^n}{\s_x(f)}$ and $\bs(f) = \max_{x\in \B^n}{\bs_x(f)}$.

\subsection{Degree measures}\label{sec:degs}

A polynomial $q \in \Reals[x_1, \ldots, x_n]$ is said to represent
the function $f:\B^n\to\B$ if
$q(x)=f(x)$ for all $x\in\B^n$.
A polynomial $q$ is said to $\eps$-approximate $f$ if $q(x)\in[0,\eps]$
for all $x\in f^{-1}(0)$ and $q(x)\in[1-\eps,1]$ for all
$x\in f^{-1}(1)$.
The degree of $f$, denoted by
$\deg(f)$, is the minimum degree of a polynomial representing $f$.
The $\eps$-approximate degree, denoted by $\tilde{\deg}_\eps(f)$,
is the minimum degree of a polynomial $\eps$-approximating $f$.
We will omit $\eps$ when $\eps=1/3$.
We know that $\D(f)\geq\deg(f)$,
$\R(f)\geq\tilde{\deg}(f)$, and $\Q(f)\geq\tilde{\deg}(f)/2$.

The degree of $f$ as a polynomial is also called the Fourier-degree of $f$, which equals $\max\{|S|:|\hat{f}(S)|\neq 0\}$ where $\hat{f}(S) := \E_x[f(x) \cdot (-1)^{\sum_{i\in S} x_i}]$. In particular, $\deg(f)<n$ if and only if $f$ agrees with the Parity function, $\Parity_n(x) = \oplus_{i=1}^{n}x_i$, on exactly half of the inputs.

\section{Degree, spectral sensitivity, and quantum query complexity}\label{sec:mainproof}

Before proving \Cref{thm:deg vs. Q}, which is based on Huang's proof, we reinterpret his result in terms of a new complexity measure of Boolean functions that we call $\lambda(f)$: the spectral norm of the sensitivity graph of $f$.

\begin{definition}[Sensitivity Graph $G_f$, Spectral Sensitivity $\lambda(f)$]\label{def:sengraph}
	Let $f:\B^n \to \B$ be a Boolean function.
	The {\sf sensitivity graph} of $f$, $G_f = (V,E)$ is a subgraph of the Boolean hypercube, where $V= \B^n$, and $E = \{(x,x\oplus e_i)\in V\times V: i\in [n], f(x)\neq f(x\oplus e_i)\}$. That is, $E$ is the set of edges between neighbors on the hypercube that have different $f$-value. Let $A_f$ be the adjacency matrix of the graph $G_f$.  We define the {\sf spectral sensitivity} of $f$ as $\lambda(f)=\norm{A_f}$.
\end{definition}

Note that since $G_f$ is bipartite, the largest and smallest eigenvalues of $A_f$ are equal in magnitude~\cite[Theorem 8.8.2]{GR01}, and because $A_f$ is a real symmetric matrix, $\lambda(f)$ is also the maximum eigenvalue of $A_f$.  

Huang's proof of the sensitivity conjecture~\cite{Huang2019} can be divided into two steps:
\begin{enumerate}
	\item $\forall{f}: \deg(f) \le \lambda(f)^2$
	\item $\forall{f}: \lambda(f) \le \s(f)$
\end{enumerate} 

The second step is the simple fact that the spectral norm of an adjacency matrix is at most the maximum degree of any vertex in the graph, which equals $\s(f)$ in this case.

We reprove the first claim of Huang's proof~\cite{Huang2019}, i.e., $\deg(f) \le \lambda(f)^2$, for completeness. This is \Cref{thm:Huang} from the introduction.

\Huang*

\begin{proof}
	Without loss of generality we can assume that $\deg(f) = n$ since otherwise we can restrict our attention to a subcube of dimension $\deg(f)$ in which the degree remains the same and the top eigenvalue is at most $\lambda(f)$. Specifically, we can choose any monomial in the polynomial representing $f$ of degree $\deg(f)$ and set all the variables not appearing in this monomial to $0$.
		
	For $f$ with $\deg(f)=n$, let $V_0 = \{x\in \B^n: f(x) = \Parity_n(x)\}$ and $V_1=\{x\in \B^n: f(x)\neq \Parity_n(x)\}$.
	By the fact that $\deg(f)=n$ we know that $|V_0|\neq |V_1|$ as otherwise $f$ would have $0$ correlation with the $n$-variate parity function, implying that $f$'s top Fourier coefficient is $0$.
	
	We also note that any edge in the hypercube that goes between $V_0$ and $V_0$ is an edge in $G_f$ since it changes the value of $f$. This holds since for such an edge, $(x,x\oplus e_i)$, we have $f(x) = \Parity_n(x)\neq \Parity_n(x \oplus e_i) = f(x \oplus e_i)$.
	Similarly, any edge in the hypercube that goes between $V_1$ and $V_1$ is an edge in $G_f$. 
	
Assume without loss of generality that $|V_0| > |V_1|$. Thus, $|V_0|\ge 2^{n-1} + 1$. We will show that there exists a nonzero vector $v'$ supported only on the entries of $V_0$, such that $\|A_f \cdot v'\| \ge \sqrt{n} \cdot \|v'\|$.

Let $G=(V,E)$ be the complete $n$-dimensional Boolean hypercube. That is, $V = \B^n$ and $E = \{(x,x\oplus e_i)\;:\;x\in \B^n, i\in [n]\}$.
Take the following signing of the edges of the Boolean hypercube, defined recursively. 
\begin{equation}
B_1 = \begin{pmatrix}0&1 \\1&0\end{pmatrix} \  \mathrm{and} \  B_i = \begin{pmatrix}B_{i-1}&I \\I&-B_{i-1}\end{pmatrix} \   \text{for $i\in \{2,\ldots,n\}$}.    
\end{equation}
This gives a new matrix $B_n \in \{-1,0,1\}^{V\times V}$ where $B_n(x,y)=0$ if and only if $x$ is not a neighbor of $y$ in the hypercube. 

Huang showed that $B_n$ has $2^{n}/2$ eigenvalues that equal $-\sqrt{n}$ and $2^{n}/2$ eigenvalues that equal $+\sqrt{n}$.
To show this, he showed that $B_n^2 = n\cdot I$ by induction on $n$ and thus all eigenvalues of $B_n$ must be either $+\sqrt{n}$ or $-\sqrt{n}$. Then, observing that the trace of $B_n$ is $0$, as all diagonal entries equal $0$, we see that we must have an equal number of $+\sqrt{n}$ and $-\sqrt{n}$ eigenvalues. 

Thus, the subspace of eigenvectors for $B_n$ with eigenvalue $\sqrt{n}$ is of dimension $2^{n}/2$. Using $|V_1|<2^n/2$, there must exists a nonzero eigenvector for $B_n$ with eigenvalue $\sqrt{n}$ that vanishes on $V_1$. 
Fix $v$ to be any such vector. 

Let $v'$  be the vector whose entries are the absolute values of the entries of $v$.
We claim that $\|A_f \cdot v'\|_2 \ge \sqrt{n}\cdot \|v'\|_2$.
To see so, note that for every $x\in V_0$ we have 
\begin{align}(A_f \cdot v')_x &= \sum_{y\sim x:f(y)\neq f(x)} v'_y = \sum_{y\sim x:y\in V_0} v'_y = \sum_{y\sim x} v'_y \nonumber \\ 
	&\ge \sum_{y\in \B^n} |B_{x,y} v_y| \ge  \left|\sum_{y\in \B^n}B_{x,y} v_y\right| = \sqrt{n}\cdot |v_x| = \sqrt{n} \cdot v'_x\;.
\end{align}
On the other hand, for $x\in V_1$ we have $(A_f \cdot v')_x = 0 = v'_x$.
Thus the norm of $A_f\cdot v'$ is at least $\sqrt{n}$ times the norm of $v'$, and hence $\lambda(f)= \|A_f\| \ge \sqrt{n} = \sqrt{\deg(f)}$. 
\end{proof}

Finally, we prove that $\lambda(f) = O(\Q(f))$. This proof goes via the spectral adversary method, $\SA(f)$, introduced by Barnum, Saks, and Szegedy~\cite{BSS03}, which has many other equivalent formulations described in~\cite{SSpalekS06}. This result also follows from the equivalent characterization of $\lambda(f)$ as $\K(f)$, a complexity measure defined by Koutsoupias~\cite{Kou93} that we describe in more detail in \Cref{sec:properties}, and the known result that $\K(f) \leq \SA(f)$ due to Laplante, Lee, and Szegedy~\cite[Theorem 5.2]{LLS06}. We provide a self-contained proof here for completeness.

\begin{definition}[Spectral Adversary method]
Let $\{D_i\}_{i\in [n]}$ and $F$ be matrices of size $\B^n \times \B^n$ with entries in $\B$ satisfying $D_i[x,y]=1$ if and only if $x_i\neq y_i$, and $F[x,y]=1$ if and only if $f(x)\neq f(y)$. Let $\Gamma$ denote a $\B^n \times \B^n$ nonnegative symmetric matrix such that $\Gamma\circ F = \Gamma$ (i.e., the nonzero entries of  $\Gamma$ are a subset of the  the nonzero entries of $F$). Then 
$\SA(f) = \max_{\Gamma} \frac{\|\Gamma\|}{\max_{i\in[n]}{\|\Gamma \circ D_i\|}}$.
\end{definition}
Barnum, Saks, and Szegedy \cite{BSS03} proved that $\Q(f) = \Omega(\SA(f))$.
\begin{lemma}\label{lem:lambdaQ}
	For all partial Boolean functions $f$, $\lambda(f) \leq \SA(f) = O(Q(f))$.
\end{lemma}
\begin{proof}
	We first prove that $\lambda(f) \leq \SA(f)$.
	Indeed, one can take $\Gamma$ to be simply the adjacency matrix of $G_f$. That is, for any $x,y\in \B^n$ put $\Gamma[x,y] = 1$ if and only if $y \sim x$ in the hypercube and $f(x)\neq f(y)$. We observe that $\|\Gamma\|=\lambda(f)$. On the other hand, for any $i\in [n]$, $\Gamma \circ D_i$ is  the restriction of the sensitive edges in direction $i$. 
	The maximum degree in the graph represented by $\Gamma \circ D_i$ is $1$ hence  $\|\Gamma \circ D_i\|$ is at most $1$.
	Thus we have 
	\begin{equation}
			\SA(f) \ge \frac{\|\Gamma\|}{\max_{i\in[n]}{\|\Gamma \circ D_i\|}} \ge \lambda(f).
	\end{equation}
	Combining this with $\Q(f) = \Omega(\SA(f))$~\cite{BSS03}, we get $\lambda(f) \leq \SA(f) = O(Q(f))$.
\end{proof}

From \Cref{lemma:Huang} and \Cref{lem:lambdaQ} we immediately get \Cref{thm:deg vs. Q}.

\section{Degree vs.\ approximate degree}
\label{sec:degree}

In this section we establish that for all total functions, spectral sensitivity is lower bounded by approximate degree. We first prove the simpler result for exact degree and then provide two proofs of the result for approximate degree.

\subsection{Spectral sensitivity lower bounds degree}

We first show the simpler result that spectral sensitivity lower bounds (exact) degree. 

\begin{theorem}\label{thm:degree}
For all total Boolean functions $f:\B^n \to \B$, $\lambda(f) \leq \deg(f)$.
\end{theorem}

We start by expressing $\lambda(f)$ in a way that allows us to relate it to a polynomial representing $f$. In the following let $H$ denote the Hadamard matrix of size $2^n \times 2^n$, defined as $H_{xy} = (-1)^{\langle x,y \rangle}2^{-n/2}$. For any function $f:\B^n \to \Reals$, we let $\diag(f)$ be the diagonal matrix that satisfies $\diag(f)_{xx}=f(x)$. We also let $X_{xx}$ be the diagonal matrix satisfying $X_{xx} = |x|$, where $|x|$ is the Hamming weight of $x$.

\begin{lemma}\label{lem:RXRX}
Let $f:\B^n \to \B$ be a total Boolean function and $g:\B^n \to \{-1,1\}$ be defined as $g=1-2f$. 
Then $\lambda(f) = \max_{v:\norm{v}=1} v^{\mathsf{T}} (RXR - X) v$, where $R = H\diag(g)H$.
\end{lemma}

\begin{proof}
Following \Cref{def:sengraph}, let the sensitivity graph of $f$ be $G_f$, its adjacency matrix be $A_f$, and its spectral sensitivity be $\lambda(f) = \left\|A_f\right\|$. 
Since $A_f$ is a symmetric matrix, we have $\lambda(f) = \max_{v:\norm{v}=1} |v^{\mathsf{T}} A_f v|$. 
Furthermore, $G_f$ is a bipartite matrix, as there are no edges between vertices of odd and even Hamming weight, which means the spectrum of $A_f$ is symmetric about $0$~\cite[Theorem 8.8.2]{GR01}. 
Thus $\lambda(f) = \max_{v:\norm{v}=1} v^{\mathsf{T}} A_f v$.

Let $A_H$ be the adjacency matrix of the hypercube graph $(V,E)$ with $V=\B^n$ and edges $(x,x\oplus e_i)$ for all $x\in\B^n$ and $i \in [n]$. Then we can express $A_f$ as
\begin{equation}\label{eq:identity on A_f and A_H}
2A_f = A_H-\diag(g)A_H\diag(g),
\end{equation}
since the $(x,y)$ entry of the right hand side is $1-g(x)g(y)$ when $(x,y)$ is an edge in the hypercube and $0$ otherwise.
Let us rewrite this expression in the basis that diagonalizes $A_H$. It is known that $H$ diagonalizes $A_H$, and $A_H = H(n\id-2X)H$, where $\id$ is the identity matrix. This is because
\begin{equation}
    (HA_HH)_{xy} = \frac{1}{2^n}\sum_{z \in \B^n}\sum_{i \in [n]} (-1)^{\langle x,z \rangle + \langle z\oplus e_i,y\rangle}
    = \frac{1}{2^n}\sum_{z \in \B^n}(-1)^{\langle x\oplus y,z \rangle} \sum_{i \in [n]} (-1)^{y_i}.
\end{equation}
The first sum is $0$ if $x\neq y$ and otherwise the expression evaluates to $n-2|x|$ showing that $HA_HH=(n\id-2X)$. Furthermore, since $H$ is an involution (i.e., $H^2=\id$), we have $H(n\id -2X)H = A_H$. Using these two identities, we have
\begin{align}
\lambda(f) = \max_{v:\norm{v}=1} v^{\mathsf{T}} A_f v &= \max_{v:\norm{v}=1} v^{\mathsf{T}} HA_fH v \label{eq:Hv}\\
&= \frac{1}{2}  \max_{v:\norm{v}=1} v^{\mathsf{T}} (H A_H H - H \diag(g) A_H \diag(g) H) v \\
&= \frac{1}{2} \max_{v:\norm{v}=1} v^{\mathsf{T}} ( n\id-2X - H \diag(g) H (n\id-2X) H \diag(g) H ) v \\
&= \max_{v:\norm{v}=1} v^{\mathsf{T}} (-X + (H \diag(g) H) X (H \diag(g) H))v\\
&= \max_{v:\norm{v}=1} v^{\mathsf{T}} (RXR-X)v,
\end{align}
where \cref{eq:Hv} uses $H^{\mathsf{T}}=H$ and $\norm{Hv}=\norm{v}$.
\end{proof}

In the proof below we use two main properties of the matrix $R$. First, $R$ is a symmetric, orthonormal matrix. 
Second, that $R_{xy}=0$ if $|x \oplus y|>\deg(g)$, where $|x \oplus y|$ is the Hamming distance between $x$ and $y$. 
Since any polynomial representing $f$ can be transformed into one representing $g$ without changing its degree, note that $\deg(g)=\deg(f)$.
The first property follows straightforwardly and we establish the second property now. Note that this lemma does not require the output of $g$ to be in $\{-1,1\}$, a fact we use in the next section.

\begin{lemma}\label{lem:sparse}
Let $g:\B^n \to \Reals$ have real degree $d$ and let $R = H\diag(g)H$. Then for all $x,y \in \B^n$, $R_{xy} = \hat{g}(x \oplus y)$, where for all $z\in\B^n$, $\hat{g}(z) = \frac{1}{2^n} \sum_{y \in \B^n} (-1)^{\langle z,y \rangle} g(y)$. Consequently, $R_{xy}=0$ if $|x \oplus y|>d$.
\end{lemma}
\begin{proof}
From the definition of $R$, we have
\begin{equation}
R_{xy} = \frac{1}{2^n} \sum_{z \in \B^n} (-1)^{\langle x, z \rangle}  (-1)^{\langle z, y \rangle} g({z})
= \frac{1}{2^n} \sum_{z \in \B^n} (-1)^{\langle x \oplus y, z \rangle} g({z})
= \hat{g}(x \oplus y).    
\end{equation}
Since $g$ can be represented by a polynomial of degree $d$, all Fourier coefficients $\hat{g}(z)$  with $|z|>d$ are $0$ and hence if $|x \oplus y|>d$ we have $R_{xy} = \hat{g}(x \oplus y) = 0$.
\end{proof}

From \Cref{lem:RXRX}, we know that $\lambda(f)$ is the maximum value of $v^{\mathsf{T}} (RXR-X) v$ over all unit vectors $v$, which can be written as 
\begin{equation}\label{eq:cibj}
    v^{\mathsf{T}} (RXR-X) v = \sum_{x \in \B^n} |x| (Rv)_x^2 - \sum_{x \in \B^n} |x|v_x^2 = \sum_{i=1}^n ic_i - \sum_{j=1}^n jb_j,
\end{equation}
where we have defined $c_i := \sum_{x:|x|=i} (Rv)_x^2$ and $b_j := \sum_{x:|x|=j} v_x^2$. 

To upper bound this expression, we need to relate the $c_i$ and $b_j$ quantities. We can establish a relationship using the fact that because $R$ is sparse, if the input vector $v$ is concentrated on one Hamming weight, then $Rv$ will be concentrated on nearby Hamming weights (up to distance $d$). 
We formalize this idea in the lemma below. Note that this lemma does not use all the properties of $R$ that we have established and we will need this in the next section.

\begin{lemma}\label{lem:cibj}
Let $R$ be a matrix with $\norm{R}\leq 1$ satisfying $R_{xy}=0$ when $|x \oplus y|>d$. For any vector $v$, define $c_i := \sum_{x:|x|=i} (Rv)_x^2$ and $b_j := \sum_{x:|x|=j} v_x^2$. Then for any $r \in \{d+1,\ldots, n\}$, we have 
\begin{equation}
    \sum_{i=r}^n c_i \leq \sum_{j=r-d}^n b_j.
\end{equation}
\end{lemma}
\begin{proof} By expanding the definition of $c_i$, we have 
\begin{equation}
    \sum_{i=r}^n c_i = \sum_{y:|y|\geq r} (Rv)_y^2 
    = \sum_{y:|y|\geq r} \left(\sum_{x \in \B^n} R_{yx}v_x \right)^2 
    = \sum_{y:|y|\geq r} \left(\sum_{x \in \B^n} R_{yx} \Pi_{(\geq r-d)} v_x \right)^2, 
\end{equation}
where we define $\Pi_{(\geq r)}$ to be the diagonal projector that satisfies the following for any vector $v$:
\begin{equation}
(\Pi_{(\geq r)} v)_x = 
    \begin{cases}
    v_x & \textrm{if} \quad |x|\geq r\\
    0 & \textrm{otherwise}
    \end{cases}.
\end{equation}
The last equality holds because $R(x,y)=0$ if $|x \oplus y|>d$ so we can restrict the sum over $x$ to be over those $x$ with $|x|\geq r-d$, and thus the only entries $v_x$ that appear in the sum have $|x|\geq r-d$.
Thus we have
\begin{equation}
    \sum_{i=r}^n c_i = \sum_{y:|y|\geq r} \left( (R(\Pi_{(\geq r-d)}v)_y \right)^2 \leq \sum_{y \in \B^n} \left( (R(\Pi_{(\geq r-d)}v)_y \right)^2 = \norm{R(\Pi_{(\geq r-d)}v)}^2, 
\end{equation}
where the inequality holds because we only added more positive numbers to the sum by relaxing the sum over $y$. Finally, 
\begin{equation}
     \norm{R(\Pi_{(\geq r-d)}v)}^2 \leq \norm{\Pi_{(\geq r-d)}v}^2  = \sum_{x:|x|\geq r-d} v_x^2 = \sum_{j=r-d}^n b_j.
\end{equation}
where the inequality uses $\norm{R}\leq 1$. This yields the claimed result.
\end{proof}

We can now formally show the upper bound on $v^{\mathsf{T}} (RXR-X) v$ for a unit vector $v$.

\begin{lemma}\label{lem:d}
Let $R$ be a matrix with $\norm{R}\leq 1$ satisfying $R_{xy}=0$ when $|x \oplus y|>d$. For any unit vector $v$, we have
\begin{equation}
    v^{\mathsf{T}} (RXR-X) v \leq d.
\end{equation}
\end{lemma}
\begin{proof}
As shown in \cref{eq:cibj}, 
\begin{equation}\label{eq:sum}
v^{\mathsf{T}} (RXR-X) v = \sum_{i=1}^n ic_i - \sum_{j=1}^n jb_j ,   
\end{equation}
where $c_i := \sum_{x:|x|=i} (Rv)_x^2$ and $b_j := \sum_{x:|x|=j} v_x^2$. \Cref{lem:cibj} already established the following inequalities for all $r \in \{d+1,\ldots,n\}$:
\begin{equation}\label{eq:sumbc}
    \sum_{i=r}^n c_i \leq \sum_{j=r-d}^n b_j,
\end{equation}
If we sum up these inequalities for all values of $r \in \{d+1,\ldots,n\}$, we don't quite get the sums that appear in \cref{eq:sum}. So let's throw in some additional inequalities.  For $r \in \{1,\ldots,d\}$, we have
\begin{equation}\label{eq:sumc}
    \sum_{i=r}^n c_i \leq 1,
\end{equation}
which uses the fact that $\sum_{i=r}^n c_i\leq \sum_{i=0}^n c_i = \norm{Rv}^2 \leq 1$ because $\norm{R}\leq 1$ and $\norm{v}=1$. We also have for all $k\in \{0,\ldots, d-1\}$
\begin{equation}\label{eq:sumb}
    0 \leq \sum_{j=n-k}^n b_j,
\end{equation}
using the fact that all $b_j\geq 0$.

If we sum up all the inequalities in \cref{eq:sumbc} for $r \in \{d+1,\ldots,n\}$, the inequalities in \cref{eq:sumc} for $r \in \{1,\ldots,d\}$, and the inequalities in \cref{eq:sumb} for $k\in \{0,\ldots, d-1\}$, we get
\begin{equation}
\sum_{i=1}^n ic_i \leq \sum_{j=1}^n jb_j + d,   
\end{equation}
which shows that $v^{\mathsf{T}} (RXR-X) v \leq d$.
\end{proof}

We can now establish \Cref{thm:degree}.

\begin{proof}[Proof of \Cref{thm:degree}]
From \Cref{lem:RXRX}, we have that $\lambda(f) = \max_{\norm{v}=1} v^{\mathsf{T}} (RXR-X) v$ and from \Cref{lem:d} we know that for any unit vector $v$, $v^{\mathsf{T}} (RXR-X) v \leq d = \deg(g) = \deg(f)$.
\end{proof}

\subsection{Spectral sensitivity lower bounds approximate degree}

We will now strengthen the previous result to show that spectral sensitivity also lower bounds approximate degree. As an intermediate result, we first establish this bound with a log factor. The stronger result without a log factor will follow from this in a completely black box way.

\begin{lemma}\label{lem:weakadeg}
For all total Boolean functions $f:\B^n \to \B$, $\lambda(f) = O(\adeg(f)\log n)$.
\end{lemma}

We use the same notation as in the previous section, where $f:\B^n \to \B$ is the total Boolean function under consideration and $g(x)=1-2f(x)$. 

Let $\tilde{g}$ be a minimum degree polynomial that $\eps$-approximates the function $g:\B^n \to \{-1, 1\}$ for some $\eps<1/2$ to be chosen later. 
We know that there exists a polynomial that $1/3$-approximates $f$ and has degree $\adeg(f)$. 
It is also known that $\adeg_\eps(f) = O(\adeg(f) \log (1/\eps))$ for any (total or partial) function $f$. See \cite{BNRdW07} for an explicit construction using ``amplification polynomials.'' 
Hence there there is a $\eps$-approximating polynomial for $f$, and hence for $g$, of degree at most $O(\adeg(f)\log (1/\eps))$. Let $\tilde{g}$ be such a polynomial with degree $d = O(\adeg(f)\log (1/\eps))$. Specifically, we have for all $x \in \B^n$ that $g(x)=-1 \implies \tilde{g}(x)\in [-1,-1+\eps]$ and $g(x)=1 \implies \tilde{g}(x)\in [1-\eps,1]$.

We start by proving a statement analogous to \Cref{lem:RXRX}, but using a polynomial that $\eps$-approximates $f$.

\begin{lemma}\label{lem:RXRXapprox}
Let $f:\B^n \to \B$ be a total Boolean function and $g:\B^n \to \{-1,1\}$ be defined as $g=1-2f$. Let $\tilde{g}$ be a degree $O(\adeg(f)\log(1/\eps))$ polynomial that $\eps$-approximates $g$.
Then $\lambda(f) = \max_{v:\norm{v}=1} v^{\mathsf{T}} (\tilde{R}X\tilde{R} - X)v + 3\eps n$, where $\tilde{R} = H\diag(\tilde{g})H$.
\end{lemma}
\begin{proof}
Using \Cref{lem:RXRX}, we have
\begin{align}
\lambda(f) &=  \max_{v:\norm{v}=1} v^{\mathsf{T}} (RXR-X)v \\
&= \max_{v:\norm{v}=1} v^{\mathsf{T}}  (H\diag(g)HXH\diag(g)H - X)v \\
&= \max_{v:\norm{v}=1} v^{\mathsf{T}} (H(\diag(\widetilde{g}) + \diag(g-\widetilde{g})) HXH (\diag(\widetilde{g}) + \diag(g-\widetilde{g}))H - X) \\
&\leq \max_{v:\norm{v}=1} v^{\mathsf{T}} (H\diag(\widetilde{g})HXH\diag(\widetilde{g})H - X)v + 3\eps n\\
&=\max_{v:\norm{v}=1} v^{\mathsf{T}} (\tR X \tR - X)v + 3\eps n,
\end{align}
where the inequality follows from the fact that $\|\diag(g-\widetilde{g}))\| \leq \eps, \|H\| = 1$, and $\|X\| = n$.
\end{proof}

As before, let us examine the matrix $\tilde{R} = H\diag(\tilde{g})H$. It follows from the definition that $\tilde{R}$ is a symmetric matrix. It is also nearly orthonormal and satisfies $\norm{\tilde{R}}\leq 1$ 
because
\begin{align}\label{eq:normR}
\norm{\tilde{R}} = \max_x |\tilde{g}(x)| \le 1. 
\end{align}

Finally, we still have $\tilde{R}_{xy}=0$ if $|x \oplus y|>d=\deg(\tilde{g})$ as before from \Cref{lem:sparse}, since the lemma did not assume that the output of $g$ was in $\{-1,1\}$. We are now ready to prove \Cref{lem:weakadeg}.

\begin{proof}[Proof of \Cref{lem:weakadeg}]
From \Cref{lem:RXRXapprox} we have that  $\lambda(f) = \max_{v:\norm{v}=1} v^{\mathsf{T}} (\tilde{R}X\tilde{R} - X)v + 3\eps n$, where $\tilde{R}=H\diag(\tilde{g})H$ and $\tilde{g}$ is an $\eps$-approximating polynomial for $g$ of degree $d=O(\adeg(f)\log(1/\eps))$. The matrix $\tR$ satisfies the assumptions of \Cref{lem:d}, so we have $\lambda(f) \leq d + 3\eps n$. Choosing $\eps = 1/3n$ gives us
$\lambda(f) = O(\adeg(f) \log n)$.
\end{proof}

Our main result now follows from this weaker statement in a black box way.

\begin{theorem}\label{thm:lambda lower bounds adeg}
For all total Boolean functions $f:\B^n \to \B$, $\lambda(f) = O(\adeg(f))$.
\end{theorem}
\begin{proof}
This proof uses Boolean function composition and the tensor power trick~\cite[1.9.4]{Tao08}. It relies on the composition properties of the complexity measures $\lambda(f)$ and $\adeg(f)$. 

First, it is not too hard to show that for all Boolean functions $f$ and $g$,
\begin{equation}
    \lambda(f \circ g) = \lambda(f) \lambda(g).
\end{equation}
This essentially follows from known results on the quantum adversary bound, but we include a proof the appendix (\Cref{thm:composition}) for completeness. We only need the $\geq$ direction in this proof. 

Second, we need the fact that approximate degree composes with at most a constant factor overhead in the upper bound direction. Sherstov~\cite{She13} showed that for all total functions $f$ and $g$, we have
\begin{equation}
    \adeg(f \circ g) \leq c \adeg(f) \adeg(g),
\end{equation}
for some universal constant $c \geq 1$.

Now from \Cref{lem:weakadeg}, we know that there exists a constant $c'$ such that for all $f:\B^n \to \B$, we have 
\begin{equation}
\lambda(f) \leq c' \adeg(f) \log n.    
\end{equation}
Let $f^k:\B^{n^k}\to\B$ denote the function $f$ composed with itself $k$ times. Then we have for all $k \in \N$,
\begin{equation}
    \lambda(f)^k = \lambda(f^k) \leq c' \adeg(f^k) \log(n^k) \leq c' c^{k-1} \adeg(f)^k (k \log n). 
\end{equation}
Taking the $k$th root on both sides gives us
\begin{equation}
    \lambda(f) \leq (c'k\log n)^{1/k} c \adeg(f).
\end{equation}
Since this equation holds for arbitrarily large $k$, we must have
\begin{equation}
    \lambda(f) \leq c \adeg(f),
\end{equation}
which completes the proof.
\end{proof}

\subsection{An alternate self-contained proof}

We now reprove the previous result, that spectral sensitivity lower bounds approximate degree, without using Sherstov's composition theorem for approximate degree~\cite{She13}. Our alternate proof also has the advantage of yielding a tighter upper bound by a constant factor. 

As in the previous section, $f:\B^n \to \B$ is the total function under consideration and we want to relate $\lambda(f)$ to $\adeg_\eps(f)$, where for any $\eps \in [0,1/2)$, $\adeg_\eps(f)$ is the minimum degree of a polynomial $q$ such that for all $x\in\B^n$, $f(x)=1 \implies q(x) \in [1-2\eps,1]$ and $f(x)=0 \implies q(x) \in [-1,-1+2\eps]$.

The main result of this section is the following, which also implies \Cref{thm:degree} by setting $\eps=0$.

\begin{theorem}\label{thm:tightadeg}
For all total Boolean functions $f:\B^n \to \B$ and $\eps \in [0,1/2)$, $\lambda(f) \leq \frac{1}{1-2\eps}\adeg_\eps(f)$.
\end{theorem}

We start by upper bounding $\lambda(f)$ by the norm of a matrix $B_q$ that is derived from the polynomial $q$. For any $\eps$, let $q$ be an $\eps$-approximating polynomial for $f$ (as defined above). We then define for all $x,y \in \B^n$, 
\begin{equation}\label{eq:Bq}
    (B_q)_{xy} = 
    \begin{cases}
    \frac{1}{2}(q(x)-q(y)) & \mathrm{if} \  |x \oplus y|=1, \\
    0 & \mathrm{if} \ |x\oplus y| \ne 1.
    \end{cases}
\end{equation}

\begin{lemma}\label{lem:lambdaBq}
For any total Boolean function $f:\B^n \to \B$ and $B_q$ as defined in \cref{eq:Bq} from an $\eps$-approximating polynomial $q$, we have $\lambda(f) \leq \frac{1}{1-2\eps} \norm{B_q}$.
\end{lemma}
\begin{proof}
We know that $\lambda(f)=\norm{A_f}$, where $(A_f)_{xy} = 1$ if and only if $|x \oplus y|=1$ and $f(x)\neq f(y)$, and otherwise $(A_f)_{xy} = 0$. Because $A_f$ has nonzero entries only when $f(x) \neq f(y)$, if we reorder the basis $\B^n$ such that all inputs with $f(x)=0$ come first and all inputs with $f(y)=1$ come after, then $A_f$ will be a block matrix of the form
\begin{equation}
    A_f = \begin{pmatrix} 0 & A^{\mathsf{T}}\\ A & 0 \end{pmatrix}.
\end{equation}
It is easy to see that $\norm{A_f}=\norm{A}$. Let us now write $B_q$ in the same reordered basis, and call the bottom left matrix $B$.  For $x \in f^{-1}(1)$ and $y \in f^{-1}(0)$, $B_{xy} = \frac{1}{2}(q(x)-q(y))$ if $|x \oplus y|=1$. For these inputs, we know that $q(x) \in [1-2\eps,1]$ and $q(y) \in [-1,-1+2\eps]$, and thus $B_{xy} \in [1-2\eps,1]$ if $|x \oplus y|=1$. All other entries in $B$ equal $0$. The matrix $A$  satisfies $A_{xy} = 1$ if $|x \oplus y|=1$. Thus we observe that $B_{xy} \geq (1-2\eps)A_{xy}$ for all $x \in f^{-1}(1)$ and $y \in f^{-1}(0)$.
Let $u$ and $v$ be unit vectors such that $u^\mathsf{T} A v=\norm{A}$.
Since $A$ is a nonnegative matrix, we may assume without loss of generality that $u,v\ge 0$.\footnote{Otherwise, by taking point-wise absolute values on $u$ and $v$ we would get two unit vectors $\tilde{u}$ and $\tilde{v}$ with $\tilde{u}^{\transpose} A \tilde{v} = \sum_{i,j} \tilde{u}_i A_{i,j} \tilde{v}_j \ge \sum_{i,j} u_i A_{i,j} v_j = u^{\transpose} A v$.
}
Using these vectors, we see that $\norm{B} \geq u^{\mathsf{T}}Bv \geq (1-2\eps)u^{\mathsf{T}}Av = (1-2\eps)\norm{A}$.

Since $B$ is a submatrix of $B_q$, we have $\norm{B_q} \geq \norm{B} \geq (1-2\eps)\norm{A} =
(1-2\eps)\norm{A_f} =(1-2\eps)\lambda(f)$.
\end{proof}

We now express $B_q$ in terms of matrices that are easier to work with.
The first matrix, which was also defined in the previous section, is 
$\tilde{R}=H\diag(q)H$, where $\diag(q)$ denotes the $2^n\times 2^n$ matrix with $q(x)$ on the diagonal.
The second matrix is $W$ of size $2^n\times 2^n$, defined as $W_{xy}=|x|-|y|$ for all $x,y\in \B^n$. 

In the following, for two $m\times n$ matrices $A$ and $B$, we denote by $A \odot B$ the Hadamard product (i.e., entrywise product) of $A$ and $B$.  $A \odot B$ is an $m\times n$ matrix, with elements given by $(A\odot B)_{i,j} = A_{i,j} B_{i,j}$.
\begin{lemma}\label{lem:BqWR}
For the matrix $B_q$ defined in \cref{eq:Bq}, and $W$ and $\tilde{R}$ defined above, $\|B_q\|=\|W\odot\tilde{R}\|$.
\end{lemma}

\begin{proof}
We have $B_q=(\diag(q)A_H-A_H\diag(q))/2$, where $A_H$ is the adjacency matrix of the Boolean hypercube. This can be verified by observing that at position $(x,y)$, both sides are $0$ if
$|x\oplus y|\ne 1$, and otherwise both sides are $(q(x)-q(y))/2$.
Now, recall that $A_H=H(n\id-2X)H$, where $\id$ is the identity matrix, $X$ is the matrix satisfying $X_{xx}=|x|$ and $H$ is the Hadamard matrix. Hence
\begin{equation}
HB_qH=\frac{H\diag(q)H(nI-2X)-(nI-2X)H\diag(q)H}{2}
=X\tilde{R}-\tilde{R}X.    
\end{equation}
Then for any $x,y \in \B^n$,
\begin{equation}
    (HB_qH)_{xy} = (X\tilde{R}-\tilde{R}X)_{xy} = |x|\tilde{R}_{xy} - |y|\tilde{R}_{xy} =(|x|-|y|)\tilde{R}_{xy} = (W \odot \tilde{R})_{xy}.
\end{equation}
Since $H$ is unitary, $\|B_q\|=\|HB_qH\|=\|W\odot \tilde{R}\|$.
\end{proof}

We now need to upper bound $\|W\odot \tilde{R}\|$.
First, we observe that~\Cref{lem:sparse} shows that the matrix $\tilde{R}$ satisfies $\tilde{R}_{xy} = 0$ if $|x \oplus y| \geq d = \adeg_\eps(f)$. The lemma was established for a different matrix $R$, but $\tilde{R}$ satisfies the assumptions of the lemma as well. 
Then, due to the above observation, $W \odot \tilde{R} = V\odot \tilde{R}$ for any matrix $V$ that satisfies $V_{xy} = W_{xy}$ for $x,y$ with $|x\oplus y|\le d$. Thus, it suffices to bound  $\| V\odot \tilde{R}\|$ for any matrix $V$ as above. In particular, we could design a matrix $V$ as above all whose entries are bounded by $d$ in absolute value. It is tempting to say that since $V$'s entries are bounded by $d$ and since $\|\tilde{R}\|\le 1$ then $\|V \odot \tilde{R}\|\le d$ however this is not well-justified since matrix norms can increase dramatically even if we just change the signs of some entries in a matrix.
Nevertheless, we would show that there exists a choice of $V$ for which the above assertion hold. Moreover we would show that for a certain choice of $V$ as above, it holds that $\|V \odot A\| \le d \|A\|$ for any matrix $A$. Indeed, such a property is captured by the $\gamma_2$-norm of $V$. This $\gamma_2$-norm arises 
in communication complexity~\cite{LMSS07,LSS08} and is also known as the
Schur product norm~\cite{Wal86}.
It is defined as follows.
\begin{definition}[$\gamma_2$ norm]\label{def:gamma}
For any $m \times n$ matrix $A$, we define 
\begin{equation}
    \gamma_2(A) 
    = \min_{X,Y: X^{\mathsf{T}}Y = A} c(X) c(Y),
\end{equation}
where $c(X)$ is the maximum $\ell_2$ norm of any column of $X$.
\end{definition}
A crucial property of the $\gamma_2$-norm is that $\|A \odot B\| \le \gamma_2(A) \cdot \|B\|$ for any two $m\times n$ matrices $A$ and $B$.%
\footnote{This inequality can also be used to define the $\gamma_2$ norm. 
As \cite[Theorem 9]{LSS08} shows, we can equivalently define $\gamma_2(A)$ as the maximum of $\norm{A \odot B}$ over all matrices $B$ with $\norm{B}\leq 1$.} We use it to prove the next lemma.

\begin{lemma}\label{lem:BqV}
Let $B_q$ and $W$ be as defined above. Let $V$ be any matrix that satisfies $V_{xy}=W_{xy}$ when $||x|-|y||\le d$. Then $\norm{B_q} \leq \gamma_2(V)$.
\end{lemma}
\begin{proof}
In \Cref{lem:BqWR} we showed that $\norm{B_q} = \norm{W \odot \tilde{R}}$. Since $\tilde{R}_{xy} = 0$ if $|x\oplus y| \geq d$, $W \odot \tilde{R} = V \odot \tilde{R}$ since $V$ and $W$ agree on inputs where $||x|-|y|| \leq d$, which is implied by the condition $|x \oplus y| \leq d$. 
Thus $\norm{B_q} = \norm{V \odot \tilde{R}}$.
We then use the relationship $\norm{A \odot B} \leq \gamma_2(A) \norm{B}$, which is not hard to show (see, e.g., \cite[Section 5]{AHJ87} or \cite[Theorem 9]{LSS08}).
This gives us $\norm{B_q} \leq \gamma_2(V) \norm{\tilde{R}}$. 
Finally, since $\tilde{R} = H \diag(q) H$ and for any $x$, $|q(x)|\leq 1$, we have $\norm{\tilde{R}} \leq 1$.
\end{proof}

Now we can upper bound $\norm{B_q}$ by $\gamma_2(V)$ for any matrix $V$ that satisfies $V_{xy}=|x|-|y|$ when $||x|-|y||\leq d$. Instead of working with $V$, which is a $2^n \times 2^n$ matrix, the following lemma will allow us work with an $(n+1) \times (n+1)$ matrix.

\begin{lemma}\label{lem:VM}
Let $M$ be any $(n+1) \times (n+1)$ matrix that satisfies $M_{st} = s-t$ for all $s,t \in \{0,\ldots, n\}$ with $|s-t|\leq d$. Then there exists a matrix $V$ satisfying the conditions of \Cref{lem:BqV}, such that $\gamma_2(V) \leq \gamma_2(M)$.
\end{lemma}

\begin{proof}
If $M$ is such a matrix, then we can get
an appropriate matrix $V$ simply by duplicating
rows and columns of $M$. That is, index the rows
and columns of $M$ by $\{0,1,\dots,n\}$,
and let $V$ be the $2^n\times 2^n$
matrix defined by $V_{x,y}=M_{|x|,|y|}$.
Then we can get $V$ from $M$ by duplicating
row $j$ of $M$ to create $\binom{n}{j}$ copies of it,
for each $j\in\{0,1,\dots,n\}$, and then repeating
the duplication for the columns.

To ensure that $\gamma_2(V)\le\gamma_2(M)$,
all we need to show is that duplicating a row
or column of a matrix does not increase
$\gamma_2(\cdot)$. This is easy to see:
if the matrix we start with is $M=X^{\mathsf{T}}Y$ with
$c(X)c(Y)=\gamma_2(M)$, then duplicating a row
of $M$ gives the matrix $M'$, which
can be factored $M'=(X')^{\mathsf{T}}Y$, where
$X'$ is the matrix we get by duplicating
the corresponding row of $X^{\mathsf{T}}$ (which is a column of $X$).
Since duplicating a column of $X$ does not affect
$c(X)$, we get a factorization of $M'$
which certifies that
$\gamma_2(M')\le c(X)c(Y)=\gamma_2(M)$.
A similar argument shows that duplicating a column
of $M$ also does not increase $\gamma_2(\cdot)$.
\end{proof}

\begin{lemma}\label{lem:M}
There is an $(n+1)\times(n+1)$ matrix $M$
such that $\gamma_2(M)\le d$ and 
$M_{st}=s-t$ for all $s,t\in\{0,1,\dots,n\}$ with
$|s-t|\le d$.
\end{lemma}

\begin{proof}
We start with a slightly informal ``picture'' proof.
To explain the matrix $M$, it will be simplest
to give an example for the case $d=3$, $n=12$.
We pick $M$ to be the following matrix:
{\small \begin{align}\left(\begin{smallmatrix*}[r]
 0 & -1 & -2 & -3 & -2 & -1 &  0 & +1 & +2 & +3 & +2 & +1 &  0 \\
+1 &  0 & -1 & -2 & -3 & -2 & -1 &  0 & +1 & +2 & +3 & +2 & +1 \\
+2 & +1 &  0 & -1 & -2 & -3 & -2 & -1 &  0 & +1 & +2 & +3 & +2 \\
+3 & +2 & +1 &  0 & -1 & -2 & -3 & -2 & -1 &  0 & +1 & +2 & +3 \\
+2 & +3 & +2 & +1 &  0 & -1 & -2 & -3 & -2 & -1 &  0 & +1 & +2 \\
+1 & +2 & +3 & +2 & +1 &  0 & -1 & -2 & -3 & -2 & -1 &  0 & +1 \\
 0 & +1 & +2 & +3 & +2 & +1 &  0 & -1 & -2 & -3 & -2 & -1 &  0 \\
-1 &  0 & +1 & +2 & +3 & +2 & +1 &  0 & -1 & -2 & -3 & -2 & -1 \\
-2 & -1 &  0 & +1 & +2 & +3 & +2 & +1 &  0 & -1 & -2 & -3 & -2 \\
-3 & -2 & -1 &  0 & +1 & +2 & +3 & +2 & +1 &  0 & -1 & -2 & -3 \\
-2 & -3 & -2 & -1 &  0 & +1 & +2 & +3 & +2 & +1 &  0 & -1 & -2 \\
-1 & -2 & -3 & -2 & -1 &  0 & +1 & +2 & +3 & +2 & +1 &  0 & -1 \\
 0 & -1 & -2 & -3 & -2 & -1 &  0 & +1 & +2 & +3 & +2 & +1 &  0 
\end{smallmatrix*}\right)\end{align}}

Here the matrix $M$ is a $13\times 13$ matrix which satisfies
$M_{st}=s-t$ when $|s-t|\le 3$: that is, within distance $3$ of
its diagonal, the entries of $M$ are equal to the distance to the
diagonal (and they are positive below the diagonal and negative
above it).
Next, we write $M$ as a sum of $d=3$ matrices
that have entries in $\{-1,0,+1\}$.
Note that each of the following $d$ matrices have $d\times d$ blocks of $-1$ just above the diagonal and $d\times d$ blocks of $+1$ just below the diagonal. 
{\small
\begin{align}\label{eq:matrix}
\left(\begin{smallmatrix*}[r]
 0 & -1 & -2 & -3 & -2 & -1 &  0 & +1 & +2 & +3 & +2 & +1 &  0 \\
+1 &  0 & -1 & -2 & -3 & -2 & -1 &  0 & +1 & +2 & +3 & +2 & +1 \\
+2 & +1 &  0 & -1 & -2 & -3 & -2 & -1 &  0 & +1 & +2 & +3 & +2 \\
+3 & +2 & +1 &  0 & -1 & -2 & -3 & -2 & -1 &  0 & +1 & +2 & +3 \\
+2 & +3 & +2 & +1 &  0 & -1 & -2 & -3 & -2 & -1 &  0 & +1 & +2 \\
+1 & +2 & +3 & +2 & +1 &  0 & -1 & -2 & -3 & -2 & -1 &  0 & +1 \\
 0 & +1 & +2 & +3 & +2 & +1 &  0 & -1 & -2 & -3 & -2 & -1 &  0 \\
-1 &  0 & +1 & +2 & +3 & +2 & +1 &  0 & -1 & -2 & -3 & -2 & -1 \\
-2 & -1 &  0 & +1 & +2 & +3 & +2 & +1 &  0 & -1 & -2 & -3 & -2 \\
-3 & -2 & -1 &  0 & +1 & +2 & +3 & +2 & +1 &  0 & -1 & -2 & -3 \\
-2 & -3 & -2 & -1 &  0 & +1 & +2 & +3 & +2 & +1 &  0 & -1 & -2 \\
-1 & -2 & -3 & -2 & -1 &  0 & +1 & +2 & +3 & +2 & +1 &  0 & -1 \\
 0 & -1 & -2 & -3 & -2 & -1 &  0 & +1 & +2 & +3 & +2 & +1 &  0 
\end{smallmatrix*}\right)
&=
\left(\begin{smallmatrix*}[r]
 0 &  0 &  0 & -1 & -1 & -1 &  0 &  0 &  0 & +1 & +1 & +1 &  0 \\
 0 &  0 &  0 & -1 & -1 & -1 &  0 &  0 &  0 & +1 & +1 & +1 &  0 \\
 0 &  0 &  0 & -1 & -1 & -1 &  0 &  0 &  0 & +1 & +1 & +1 &  0 \\
+1 & +1 & +1 &  0 &  0 &  0 & -1 & -1 & -1 &  0 &  0 &  0 & +1 \\
+1 & +1 & +1 &  0 &  0 &  0 & -1 & -1 & -1 &  0 &  0 &  0 & +1 \\
+1 & +1 & +1 &  0 &  0 &  0 & -1 & -1 & -1 &  0 &  0 &  0 & +1 \\
 0 &  0 &  0 & +1 & +1 & +1 &  0 &  0 &  0 & -1 & -1 & -1 &  0 \\
 0 &  0 &  0 & +1 & +1 & +1 &  0 &  0 &  0 & -1 & -1 & -1 &  0 \\
 0 &  0 &  0 & +1 & +1 & +1 &  0 &  0 &  0 & -1 & -1 & -1 &  0 \\
-1 & -1 & -1 &  0 &  0 &  0 & +1 & +1 & +1 &  0 &  0 &  0 & -1 \\
-1 & -1 & -1 &  0 &  0 &  0 & +1 & +1 & +1 &  0 &  0 &  0 & -1 \\
-1 & -1 & -1 &  0 &  0 &  0 & +1 & +1 & +1 &  0 &  0 &  0 & -1 \\
 0 &  0 &  0 & -1 & -1 & -1 &  0 &  0 &  0 & +1 & +1 & +1 &  0 
\end{smallmatrix*}\right)
 \nonumber \\
&\mkern-260mu
+
\left(\begin{smallmatrix*}[r]
 0 &  0 & -1 & -1 & -1 &  0 &  0 &  0 & +1 & +1 & +1 &  0 &  0 \\
 0 &  0 & -1 & -1 & -1 &  0 &  0 &  0 & +1 & +1 & +1 &  0 &  0 \\
+1 & +1 &  0 &  0 &  0 & -1 & -1 & -1 &  0 &  0 &  0 & +1 & +1 \\
+1 & +1 &  0 &  0 &  0 & -1 & -1 & -1 &  0 &  0 &  0 & +1 & +1 \\
+1 & +1 &  0 &  0 &  0 & -1 & -1 & -1 &  0 &  0 &  0 & +1 & +1 \\
 0 &  0 & +1 & +1 & +1 &  0 &  0 &  0 & -1 & -1 & -1 &  0 &  0 \\
 0 &  0 & +1 & +1 & +1 &  0 &  0 &  0 & -1 & -1 & -1 &  0 &  0 \\
 0 &  0 & +1 & +1 & +1 &  0 &  0 &  0 & -1 & -1 & -1 &  0 &  0 \\
-1 & -1 &  0 &  0 &  0 & +1 & +1 & +1 &  0 &  0 &  0 & -1 & -1 \\
-1 & -1 &  0 &  0 &  0 & +1 & +1 & +1 &  0 &  0 &  0 & -1 & -1 \\
-1 & -1 &  0 &  0 &  0 & +1 & +1 & +1 &  0 &  0 &  0 & -1 & -1 \\
 0 &  0 & -1 & -1 & -1 &  0 &  0 &  0 & +1 & +1 & +1 &  0 &  0 \\
 0 &  0 & -1 & -1 & -1 &  0 &  0 &  0 & +1 & +1 & +1 &  0 &  0
\end{smallmatrix*}\right)+
\left(\begin{smallmatrix*}[r]
 0 & -1 & -1 & -1 &  0 &  0 &  0 & +1 & +1 & +1 &  0 &  0 &  0 \\
+1 &  0 &  0 &  0 & -1 & -1 & -1 &  0 &  0 &  0 & +1 & +1 & +1 \\
+1 &  0 &  0 &  0 & -1 & -1 & -1 &  0 &  0 &  0 & +1 & +1 & +1 \\
+1 &  0 &  0 &  0 & -1 & -1 & -1 &  0 &  0 &  0 & +1 & +1 & +1 \\
 0 & +1 & +1 & +1 &  0 &  0 &  0 & -1 & -1 & -1 &  0 &  0 &  0 \\
 0 & +1 & +1 & +1 &  0 &  0 &  0 & -1 & -1 & -1 &  0 &  0 &  0 \\
 0 & +1 & +1 & +1 &  0 &  0 &  0 & -1 & -1 & -1 &  0 &  0 &  0 \\
-1 &  0 &  0 &  0 & +1 & +1 & +1 &  0 &  0 &  0 & -1 & -1 & -1 \\
-1 &  0 &  0 &  0 & +1 & +1 & +1 &  0 &  0 &  0 & -1 & -1 & -1 \\
-1 &  0 &  0 &  0 & +1 & +1 & +1 &  0 &  0 &  0 & -1 & -1 & -1 \\
 0 & -1 & -1 & -1 &  0 &  0 &  0 & +1 & +1 & +1 &  0 &  0 &  0 \\
 0 & -1 & -1 & -1 &  0 &  0 &  0 & +1 & +1 & +1 &  0 &  0 &  0 \\
 0 & -1 & -1 & -1 &  0 &  0 &  0 & +1 & +1 & +1 &  0 &  0 &  0
\end{smallmatrix*}\right)
\end{align}
}

For a general $d$, $M$ will be a sum of $d$ matrices that we call $C_0,C_1,\dots, C_{d-1}$. We now explain how to upper bound $\gamma_2(C_k)$ and obtain an upper bound on $\gamma_2(M)$. Define $Q$ to be the matrix \begin{equation}
Q=J_{\lceil n/4d\rceil+1}\otimes
\left(\begin{smallmatrix*}[r]
 0 & -1 &  0 & +1 \\
+1 &  0 & -1 &  0 \\
 0 & +1 &  0 & -1 \\
-1 &  0 & +1 &  0
\end{smallmatrix*}\right)
\otimes J_d,
\end{equation}
where $\otimes$ denotes the Kronecker (tensor) product,
and where $J_m$ denotes the $m\times m$ all-ones matrix.
Further, this middle matrix is
$\big(\begin{smallmatrix*}[r]
+1 & -1 \\
-1 & +1
\end{smallmatrix*}\big)
\otimes
\big(\begin{smallmatrix*}[r]
0 & -1 \\
+1 & 0
\end{smallmatrix*}\big)$.

For $k = 0, \ldots, d-1$, we let $C_{k}$ be the $n \times n$ submatrix defined by $(C_{k})_{i, j} = Q_{i+k, j+k}$.
In other words, $C_0$ is the top left $n \times n$ corner of $Q$, and for each $k\ge1$, we obtain $C_{k}$ from $C_{k-1}$ by shifting the submatrix diagonally.

Thus, the $(i, j)$ entry of $C_0+\dots+C_{d-1}$ is equal to $Q_{i, j}+Q_{i+1, j+1}+\cdots+Q_{i+d-1, j+d-1}$.
Because $Q$ is a block matrix whose blocks are of size $d \times d$, this is equal in magnitude to $|i-j|$, as $d-|i-j|$ is the length of the intersection of the line segment from $(i, j)$ to $(i+d-1, j+d-1)$ and the $d \times d$ blocks along the main diagonal.
Finally, whether the value is positive or negative depends whether the line segment additionally intersects a block of $1$'s (if $i > j$) or a block of $-1$'s (if $i < j$,) and thus $C_0+\dots+C_{d-1}$ is in fact equal to $M$.

Now that we have argued that all the matrices $C_k$ are submatrices of shifted versions of $Q$, let us compute $\gamma_2(Q)$. We now claim that $\gamma_2(Q)=1$. This follows from a few
easy-to-verify facts about the $\gamma_2(\cdot)$ norm:
\begin{enumerate}
\item $\gamma_2(A\otimes B)=\gamma_2(A)\gamma_2(B)$.
\item  $\gamma_2(J_m)=1$ for all $m$.
\item $\gamma_2(\big(\begin{smallmatrix*}[r]
+1 & -1 \\
-1 & +1
\end{smallmatrix*}\big)) = 1$ and $\gamma_2(
\big(\begin{smallmatrix*}[r]
0 & -1 \\
+1 & 0
\end{smallmatrix*}\big))=1$.
\end{enumerate}
Since $Q$ decomposes into a Kronecker product of matrices
with $\gamma_2(\cdot)=1$ we get that $\gamma_2(Q)=1$.
Finally, we use two additional properties of $\gamma_2(\cdot)$:
first, that it is subadditive (indeed, it is a norm),
so $\gamma_2(M)\le \sum_{k=0}^{d-1} \gamma_2(C_k)$;
and second, that it is non-increasing under restriction
to a submatrix, so $\gamma_2(C_k)\le\gamma_2(Q)$ for each
$k =0, 1, \ldots, d-1$. Together, these properties
imply that $\gamma_2(M)\le d$, completing the picture proof.

We also provide a more explicit and formal way of showing
that there is an appropriate matrix $M$ satisfying
$\gamma_2(M)\le d$; this method
also avoids using any properties of $\gamma_2(\cdot)$
by directly giving a factorization
$M=S^{\mathsf{T}}T$. This factorization still corresponds
to writing $M$ as the sum of $C_0+C_1+\dots +C_{d-1}$,
together with the observation that each of the matrices
$C_k$ has rank $2$ (as can be seen from their Kronecker product
decomposition).

Specifically, the matrices $S$ and $T$ will
have dimensions $2d\times (n+1)$, and will be defined as follows.
For any $s\in\{0,1,\dots, n\}$ and any $j\in\{0,\dots, 2d-1\}$,
we can write $s+j=a+b(2d)$ for unique integers
$a\in\{0,1,\dots,2d-1\}$ and $b\in\mathbb{N}$;
that is, $a$ is the remainder of $s+j$ modulo $2d$,
and $b=\lfloor(s+j)/2d\rfloor$.
If $a\in[0,d-1]$, we define $S_{js}=(-1)^b$ and $T_{js}=0$.
Otherwise (i.e., if $a\in[d,2d-1]$), we define
$S_{js}=0$ and $T_{js}=-(-1)^b$.
We will use $S_j$ to denote row $j$ of $S$, and similarly
for $T_j$.
(For reference, in terms of the picture proof, we will have
$S_j^{\mathsf{T}}T_j+S_{j+d}^{\mathsf{T}}T_{j+d}=C_j$ for all $j \in \{0,1,\ldots, d-1\}$.)

In each column of $S$ and $T$, half the entries are zero
and half are $\pm1$, and hence we have
$c(S)=c(T)=\sqrt{d}$. To
complete the proof of the lemma, all we need to do
is to set $M=S^{\mathsf{T}}T$ and to show that $M_{st}=s-t$
whenever $|s-t|\le d$.
To this end, observe that
$(S_j^{\mathsf{T}}T_j)_{st}$ is nonzero if and only if
$s+j \mod 2d \in [0,d-1]$ and $t+j\mod 2d\in [d,2d-1]$.
When $|s-t|\le d$, it's easy to see that the number
of different values of $j\in\{0,1,\dots, 2d-1\}$ for which
this happens is exactly $|s-t|$.
Moreover, if $s<t$,
then it's not hard to see that $(S_j^{\mathsf{T}}T_j)_{st}$
will be negative if it is nonzero, while if $s>t$,
it will be positive if it is nonzero.
Hence $(S^{\mathsf{T}}T)_{st}=\sum_{j=0}^{2d-1} (S_j^{\mathsf{T}}T_j)_{st}=s-t$
whenever $|s-t|\le d$.
\end{proof}

Finally, we can put the pieces together to prove \Cref{thm:tightadeg}.

\begin{proof}[Proof of \Cref{thm:tightadeg}]
From \Cref{lem:lambdaBq}, we have $\lambda(f) \leq \frac{1}{1-2\eps}\norm{B_q}$. \Cref{lem:BqV} and \Cref{lem:VM} then give us $\norm{B_q} \leq \gamma_2(M)$, for any matrix $M$ that satisfies the conditions in \Cref{lem:M}. \Cref{lem:M} constructs such a matrix $M$ with $\gamma_2(M) \leq d= \adeg_\eps(f)$, completing the proof.
\end{proof}

\section{Monotone graph properties}
\label{sec:AKR}

The Aanderaa--Karp--Rosenberg conjectures are a collection of conjectures related to the query complexity of deciding whether an input graph specified by its adjacency matrix satisfies a given property in various models of computation. 

Specifically, let the input be an $n$-vertex undirected simple graph specified by its adjacency matrix. This means we can query any unordered pair $\{i,j\}$, where $i,j\in[n]$, and learn whether there is an edge between vertex $i$ and $j$. Note that the input size is $\binom{n}{2}=\Theta(n^2)$. 

A function $f$ on $\binom{n}{2}$ variables is a graph property if it treats the input as a graph and not merely a string of length $\binom{n}{2}$. Specifically, the function must be invariant under permuting vertices of the graph. 
In other words, the function can only depend on the isomorphism class of the graph, not the specific labels of the vertices. 
A function $f$ is monotone (increasing) if for all $x,y\in\B^n$, $x \leq y \implies f(x) \leq f(y)$, where $x \leq y$ means $x_i \leq y_i$ for all $i\in[n]$. 
For a monotone function, negating a $0$ in the input cannot change the function value from $1$ to $0$. 
In the context of graph properties, if the input graph has a certain monotone graph property, then adding more edges cannot destroy the property. 

Examples of monotone graph properties include ``$G$ is connected,'' ``$G$ contains a clique of size $k$,'' ``$G$ contains a Hamiltonian cycle,'' ``$G$ has chromatic number greater than $k$,'' ``$G$ is not planar'', and ``$G$ has diameter at most $k$.'' Many commonly encountered graph properties (or their negation) are monotone graph properties. Finally, we say a function $f:\B^n\to\B$ is nontrivial if there exist inputs $x$ and $y$ such that $f(x)\neq f(y)$.

The deterministic Aanderaa--Karp--Rosenberg conjecture, also called the \emph{evasiveness conjecture},\footnote{A function $f$ is called \emph{evasive} if its deterministic query complexity equals its input size.} states that for all nontrivial monotone graph properties $f$, $\D(f) = \binom{n}{2}$. This conjecture remains open to this day, although the weaker claim that $\D(f)=\Omega(n^2)$ was proved over 40 years ago by Rivest and Vuillemin~\cite{RV76}. Several works have improved on the constant in their lower bound, and the best current result is due to Scheidweiler and Triesch~\cite{ST13}, who prove a lower bound of $\D(f)\geq (1/3-o(1)) \cdot n^2$. The evasiveness conjecture has been established in several special cases including when $n$ is prime~\cite{KSS84} and when restricted to bipartite graphs~\cite{Yao88}.

The randomized Aanderaa--Karp--Rosenberg conjecture asserts that all nontrivial monotone graph properties $f$ satisfy $\R(f) = \Omega(n^2)$. A sequence of increasingly stronger lower bounds, starting with a lower bound of $\Omega(n\log^{1/12}n)$ due to Yao~\cite{Yao91}, a lower bound of $\Omega(n^{5/4})$ due to King~\cite{Kin88}, and a lower bound of $\Omega(n^{4/3})$ due to Hajnal~\cite{Haj91}, has led to the current best lower bound of $\Omega(n^{4/3}\log^{1/3}n)$ due to Chakrabarti and Khot~\cite{CK01}. There are also two lower bounds due to  Friedgut, Kahn, and Wigderson~\cite{FKW02} and O'Donnell, Saks,  Schramm, and Servedio~\cite{OSSS05} that are better than this bound for some graph properties.

The quantum Aanderaa--Karp--Rosenberg conjecture states that all nontrivial monotone graph properties $f$ satisfy $\Q(f) = \Omega(n)$. This is the best lower bound one could hope to prove since there exist properties with $\Q(f) = O(n)$, such as the property of containing at least one edge. In fact, for any $\alpha \in [1,2]$ it is possible to construct a graph property with quantum query complexity $\Theta(n^{\alpha})$ using known lower bounds for the threshold function~\cite{BBCMW01}.

As stated in the introduction, the question was first raised by Buhrman, Cleve, de Wolf, and Zalka~\cite{BCdWZ99}, who showed a lower bound of $\Omega(\sqrt{n})$. This was improved by Yao to $\Omega(n^{2/3}\log^{1/6}n)$ using the technique in \cite{CK01} and Ambainis' adversary bound~\cite{Amb02}. Better lower bounds are known in some special cases, such as when the property is a subgraph isomorphism property, where we know a lower bound of $\Omega(n^{3/4})$ due to Kulkarni and Podder~\cite{KP16}.

As stated in \Cref{thm:qAKR}, we resolve the quantum Aanderaa--Karp--Rosenberg conjecture and show an optimal $\Omega(n)$ lower bound. The proof combines \Cref{thm:deg vs. Q} with a quadratic lower bound on the degree of nontrivial monotone graph properties. With some work, the original quadratic lower bound on the deterministic query complexity of nontrivial monotone graph properties by Rivest and Vuillemin~\cite{RV76} can be modified to prove a similar lower bound for degree. 
We were not able to find such a proof in the literature, and instead combine the following two claims to obtain the desired claim.

First, we use the result of Dodis and Khanna~\cite[Theorem 2]{DK99}:
\begin{theorem}
For all nontrivial monotone graph properties, $\deg_2(f) = \Omega(n^2)$.
\end{theorem}

Here $\deg_2(f)$ is the minimum degree of a Boolean function when represented as a polynomial over the finite field with two elements, $\mathbb{F}_2$. We combine this with a standard lemma that shows that this measure lower bounds $\deg(f)$. A proof can be found in~\cite[Proposition 6.23]{ODo09}:
\begin{lemma}
For all Boolean functions $f:\B^n \to \B$, we have $\deg_2(f) \leq \deg(f)$.
\end{lemma}

Combining these with \Cref{thm:deg vs. Q}, we get that all nontrivial monotone graph properties $f$ satisfy $\Q(f) = \Omega(n)$, which is the statement of \Cref{thm:qAKR}.

\section{Approximate degree of read-once formulas}
\label{sec:readonce}

A read-once formula over the De Morgan basis, which consists of AND gates, OR gates, and NOT gates, is a formula in which each variable appears exactly once. Examples of read-once formulas include the $\AND_n$ and $\OR_n$ functions themselves, and compositions of these functions such as $\AND_n \circ \OR_n$. 

While it was already established in \cite{NS94} that $\adeg(\AND_n) = \adeg(\OR_n) = \Theta(\sqrt{n})$, the approximate degree of $\AND_{\sqrt{n}} \circ \OR_{\sqrt{n}}$ remained open until 2013 when it was shown that $\AND_{\sqrt{n}} \circ \OR_{\sqrt{n}} = \Theta(\sqrt{n})$~\cite{BT13,She13a} using a linear programming characterization of approximate degree. 
Notice that in both cases the approximate degree is the square root of the number of variables. 
This was later extended to constant-depth balanced read-once formulas~\cite{BT15} and constant-depth unbalanced read-once formulas~\cite{BBGK18}.
We finally resolve this question for all read-once formulas by establishing \Cref{thm:readonce} from the introduction:

\readonce*

Note that we already knew that for read-once formulas $f$, $\Q(f) = \Theta(\sqrt{n})$. The lower bound was established by Barnum and Saks~\cite{BS04} and the upper bound was established by Reichardt~\cite{Rei11}.

The upper bound in \Cref{thm:readonce} follows straightforwardly from Reichardt's upper bound~\cite{Rei11}, since approximate degree lower bounds quantum query complexity~\cite{BBCMW01}. The lower bound is a consequence of \Cref{thm:degadeg} because the degree of a read-once formula is equal to the number of variables.

\begin{lemma}
For any read-once formula $f:\B^n \to \B$, we have $\deg(f) = n$.
\end{lemma}
\begin{proof}
By using De Morgan's laws, we can assume that the read-once formula only contains AND and NOT gates. 
The base case of a formula with $n=1$ variable is easy, since the only such formulas as $x_1$ and $\overline{x_1}$, which have degree $1$.
More generally, $\deg(\bar{f}) = \deg(f)$, since if a polynomial $p(x)$ equals $f(x)$ for all $x\in \B^n$, then the polynomial $1-p$ equals $\bar{f}$. 

All that remains to be shown is that for read-once formulas $f:\B^n \to \B$ and $g:\B^m \to \B$, we have $\deg(f \wedge g) = \deg(f) + \deg(g)$, where $f \wedge g:\B^{n+m} \to \B$ is the function that evaluates to $f(x) \wedge g(y)$ for $x\in\B^n$ and $y \in \B^m$. The upper bound is obvious since multiplying the polynomials that represent $f$ and $g$ gives us a polynomial for $f \wedge g$ with degree equal to the sum of their degrees. However, since the polynomial representation of a Boolean function is unique, and there is no way of cancelling out higher degree terms by multiplying these polynomials (since the polynomials involve different sets of variables), we get that $\deg(f \wedge g) = \deg(f) + \deg(g)$. Induction on the structure of the formula completes the proof.
\end{proof}

\section{Open questions}\label{sec:open}

We saw that $\lambda(f)$ lower-bounds all the complexity measures in \Cref{fig:rel}, and is polynomially related to all of them. We know that $\deg(f) \leq \lambda(f)^2$ and $\s(f) \leq \lambda(f)^2$, and these relationships are optimal, but the optimal relationships between all other complexity measures and $\lambda(f)$ remain open. For example, perhaps $\bs(f) =O(\lambda(f)^2)$ or even $\RC(f) = O(\lambda(f)^2)$? The best relationship between block sensitivity and sensitivity also remains open.

It may also be possible to relate $\lambda(f)$ to the rational degree of $f$, which is the minimum degree of polynomials $p$ and $q$ such that for all $x\in \B^n$, $q(x) \neq 0$ and $f(x)=p(x)/q(x)$. It is unknown if rational degree is polynomially related to the complexity measures in \Cref{fig:rel}, although the question has been open for a long time~\cite{NS94}.

Another longstanding open problem is to show a quadratic relation between deterministic query complexity and block sensitivity:
\begin{conjecture}
    For all Boolean functions $f:\B^n \to \B$, we have $\D(f) = O(\bs(f)^2)$.
\end{conjecture}
If this conjecture were true, it would optimally resolve several relationships in \Cref{tab:sep}, and would imply, for example, $\D(f) = O(\R(f)^2)$.

After settling the best relation between $\D(f)$ and $\Q(f)$, the next pressing question is to settle the best relation between  $\R(f)$ and $\Q(f)$.
Recently, two independent works~\cite{BS20,SSW20} showed a power $3$ separation between $\R(f)$ and $\Q(f)$, while the best known relationship is a power $4$ relationship (from this work). We conjecture that the upper bound can be improved.

\begin{conjecture}
    For all Boolean functions $f:\B^n \to \B$, we have $\R(f) = O(\Q(f)^3)$.
\end{conjecture}

Finally, for the special case of {\em monotone} total Boolean functions $f$, Beals et al.~\cite{BBCMW01} already showed in 1998 that $\D(f)=O(\Q(f)^4)$. It would be interesting to know whether this can be improved, perhaps all the way to $\D(f)=O(\Q(f)^2)$.

\section*{Acknowledgements}

We thank Troy Lee for comments on the first version of his paper and for bringing to our attention Theorem 5.2 in \cite{LLS06}. We also thank Evgeny Vihrov for finding a bug in the argument in \Cref{lem:averagesen}, which has been fixed. We thank Alex Yu for sharing with us his proof of \Cref{thm:composition}.

\bibliographystyle{alphaurl}
\bibliography{bibs}

\appendix

\section{Properties of the measure \texorpdfstring{$\lambda(f)$}{lambda}}
\label{sec:properties}

We show that the measure $\lambda(f)$ satisfies various elegant
properties. First, it can be defined in multiple ways,
one of which was introduced by Koutsoupias
back in 1993 \cite{Kou93}. It also has a formulation
as a special case of the quantum adversary bound and hence can 
be expressed as as a semidefinite program closely related to that
of the quantum adversary bound. Due to this characterization, $\lambda(f)$
can be viewed as both a maximization problem and a minimization
problem. These equivalent formulations are described in \Cref{sec:equivalent}.

Second, we show that $\lambda(f)$ satisfies perfect composition:
$\lambda(f\circ g)=\lambda(f)\lambda(g)$. 
Third,
we show that $\lambda(f)\le\sqrt{\s_0(f)\s_1(f)}$, 
which was already observed by Laplante, Lee, and Szegedy~\cite{LLS06}
(though we give a slightly different proof).
Finally, we show lower bounds on $\lambda(f)$ and an optimal quadratic separation between $\lambda(f)$ and $\s(f)$.

\subsection{Equivalent formulations}\label{sec:equivalent}

\begin{theorem}\label{thm:equivalent}
For all Boolean functions $f\colon\B^n\to\B$, we have
\begin{equation}
\lambda(f)=\K(f)=\Adv_1(f)=\Adv_1^{\pm}(f),    
\end{equation}
where the measures $\K(f)$, $\Adv_1(f)$, and $\Adv_1^{\pm}(f)$
are defined below. Furthermore, $\Adv_1(f)$ itself
has several equivalent formulations:
$\Adv_1(f)\coloneqq\SA_1(f)=\SWA_1(f)=\MM_1(f)=\GSA_1(f)$.
\end{theorem}

We now define all these measures before proving this theorem.

\paragraph{Koutsoupias complexity $\K(f)$.}
For a Boolean function $f$, let $A\subseteq f^{-1}(0)$,
and let $B\subseteq f^{-1}(1)$. Let $Q$ be the matrix
with rows and columns labeled by $A$ and $B$ respectively,
with $Q[x,y]=1$ if the Hamming distance of $x$ and $y$ is $1$,
and $Q[x,y]=0$ otherwise. Koutsoupias \cite{Kou93}
observed that $\|Q\|^2$ is a lower bound on formula size,
for every such choice of $A$ and $B$. We define
$\K(f)$ to be the maximum value of $\|Q\|$ over choices
of $A$ and $B$.
Thus $\K(f)^2$ is a lower bound on the formula size of $f$.

\paragraph{Single-bit positive adversary $\Adv_1(f)$.}
We define $\Adv_1(f)$ as a version of the adversary bound
where we are only allowed to put nonzero weight on input pairs
$(x,y)$ where $f(x)\ne f(y)$ and the Hamming distance between
$x$ and $y$ is exactly $1$. We will define $\Adv_1(f)$
in terms of the spectral adversary version, which we
also denote by $\SA_1(f)$.
$\Adv_1(f) = \SA_1(f)$ is defined as the maximum of
\begin{equation}
\frac{\|\Gamma\|}{\max_{i\in[n]}\|\Gamma\circ D_i\|}    
\end{equation}
over matrices $\Gamma$ of a special form. We require $\Gamma$
satisfy the following: (1) its entries are nonnegative reals;
(2) its rows and columns are indexed by $\Dom(f)$;
(3) $\Gamma[x,y]=0$ whenever $f(x)=f(y)$; 
(4) $\Gamma[x,y]=0$ whenever the Hamming distance of $x$ and $y$
is not $1$; and (5) $\Gamma$ is not all $0$.
In the above expression, $\circ$ refers to the Hadamard
(entrywise) product, $\Dom(f)$ is the domain of $f$, 
and $D_i$ is the $\B$-valued matrix with $D_i[x,y]=1$
if and only if $x_i\ne y_i$.

\paragraph{Single-bit negative adversary $\Adv_1^{\pm}(f)$.}
We define $\Adv_1^{\pm}(f)$ using the same definition as
$\Adv_1(f)$ above, except that the matrix $\Gamma$ is allowed
to have negative entries. Note that since
this is a relaxation of the
conditions on $\Gamma$, we clearly have
$\Adv_1^\pm(f)\ge\Adv_1(f)$.

\paragraph{Single-bit strong weighted adversary $\SWA_1(f)$.}
We define $\SWA_1(f)$ as a single-bit version of the
strong weighted adversary method $\SWA(f)$ from \cite{SSpalekS06}.
For this definition, we say a weight function
$w\colon\Dom(f)\times\Dom(f)\to[0,\infty)$ is feasible
if it is symmetric (i.e., $w(x,y)=w(y,x)$) and
if it satisfies the conditions on $\Gamma$ above
(i.e., it places weight $0$ on a pair $(x,y)$
unless both $f(x)\ne f(y)$ and the Hamming distance
between $x$ and $y$ is $1$).
We view such a feasible weight scheme $w$ as the weights
on a weighted bipartite graph, where the left vertex set is
$f^{-1}(0)$ and the right vertex set is $f^{-1}(1)$.
We let $wt(x)\coloneqq\sum_y w(x,y)$ denote the weighted degree of
$x$ in this graph, i.e., the sum of the weights of its incident
edges. Then $\SWA_1(f)$ is defined as the maximum, over
such feasible weight schemes $w$, of
\begin{equation}
\min_{x,i:w(x,x^i)>0}\frac{\sqrt{wt(x)wt(x^i)}}{w(x,x^i)}.    
\end{equation}
Here $x$ ranges over $\Dom(f)$, $i$ ranges over $[n]$, and
$x^i$ denotes the string $x$ with bit $i$ flipped.\footnote{%
Readers familiar with the adversary bound should note that
this definition is analogous a weighted version
of Ambainis's original adversary method; in the original method,
the denominator was the geometric mean of (a) the weight
of the neighbors of $x$ with disagree with $x$ at $i$,
and (b) the weight of the neighbors of $x^i$ which disagree
with $x^i$ at $i$; but in our case, both (a) and (b) are simply
$w(x,x^i)$, since $x^i$ is the only string that disagrees
with $x$ on bit $i$ and is connected to $x$ in the
bipartite graph.}

\paragraph{Single-bit minimax adversary $\MM_1(f)$.}
Unlike the other forms, we define $\MM_1(f)$ as a minimization
problem rather than a maximization problem. We
say a weight function $w\colon\Dom(f)\times[n]\to[0,\infty)$
is feasible if for all $x,y\in\Dom(f)$ with
$f(x)\ne f(y)$ and Hamming distance $1$, we have
$w(x,i)w(y,i)\ge 1$, where $i$ is the bit on which
$x$ and $y$ disagree. $\MM_1(f)$ is defined as the minimum,
over such feasible weight schemes $w$, of
\begin{equation}
    \max_{x\in\Dom(f)}\sum_{i\in[n]} w(x,i).
\end{equation}

\paragraph{Semidefinite program version $\GSA_1(f)$.}
We define $\GSA_1(f)$ to be the optimal value of the following
semidefinite program.
\begin{equation}
\begin{array}{lll}
\text{maximize}  &  \langle Z,A_f\rangle &\\
\text{subject to}& \Delta\mbox{ is diagonal} & \\
                & \tr\Delta = 1 & \\
                & \Delta-Z\circ D_i\succeq 0 & \forall i\in[n]\\
                & Z \ge 0 &
\end{array}
\end{equation}
Here $Z$ and $\Delta$ are variable matrices with rows
and columns indexed by $\Dom(f)$, $A_f$ is the $\B$-matrix with
$A_f[x,y]=1$ if and only if both $f(x)\ne f(y)$
and $(x,y)$ have Hamming distance $1$, and $D_i$
is the $\B$-matrix with $D_i[x,i]=1$ if and only if $x_i\ne y_i$.

We now prove \Cref{thm:equivalent}.

\begin{proof}
Recall that in the definition of $\K(f)$,
we picked $A\subseteq f^{-1}(0)$ and $B\subseteq f^{-1}(1)$
and defined the resulting matrix $Q$. Since the spectral norm
of a submatrix is always smaller than or equal to the spectral
norm of the original matrix, we can always assume without
loss of generality that $A=f^{-1}(0)$ and $B=f^{-1}(1)$.
Then $\K(f)=\|Q\|$ for the resulting matrix $Q$ with rows
and columns indexed by $f^{-1}(1)$ and $f^{-1}(0)$ respectively.
Now, recall that $A_f$ was the adjacency matrix of the graph
$G_f$, which has an edge between $x$ and $y$ if $f(x)\ne f(y)$
and the Hamming distance between $x$ and $y$ is $1$.
The rows and columns of $A_f$ are each indexed by $\Dom(f)$.
By rearranging them, we can make $A_f$ be block diagonal
with blocks equal to $Q$ and $Q^\dagger$. From there it
follows that $\|A_f\|=\|Q\|$, so $\lambda(f)=\K(f)$.

Next, recall that $\Adv_1(f)$ is defined as the maximum
ratio $\|\Gamma\|/\max_i\|\Gamma\circ D_i\|$ over valid
choices of $\Gamma$. Note that since $\Gamma[x,y]$
can only be nonzero if $x$ and $y$ disagree on one bit,
$\Gamma\circ D_i$ is nonzero only on pairs $(x,y)$
which disagree exactly on bit $i$. In other words,
if $P_i$ denotes the $\B$-valued matrix with $P_i[x,y]=1$
if and only if $x$ and $y$ disagree on bit $i$ and only on $i$,
then $\Gamma\circ D_i$ is nonzero only in entries where $P_i$
is $1$. Now, note that $P_i$ is a permutation matrix.
Hence, by rearranging the rows and columns of $\Gamma\circ D_i$,
we can get it to be diagonal. This means $\|\Gamma\circ D_i\|$
is the maximum entry of $\Gamma\circ D_i$, and hence
$\max_i\|\Gamma\circ D_i\|$ is the maximum entry of $\Gamma$.
It follows that $\Adv_1(f)$ is the maximum of $\|\Gamma\|$
over feasible matrices $\Gamma$ with $\max(\Gamma)\le 1$,
where $\max(\Gamma)=\max_{ij}|\Gamma_{ij}|$. This argument also holds for
$\Adv_1^\pm(f)$, which is the maximum of $\|\Gamma\|$
over feasible (possibly negative) matrices $\Gamma$
with $\max(\Gamma)\le 1$.

Next, observe that negative weights never help for maximizing
$\|\Gamma\|$: indeed, if we had $\Gamma$ with negative entries
maximizing $\|\Gamma\|$, then we would have vectors $u$ and $v$
with $\|u\|_2=\|v\|_2=1$ and $u^{\mathsf{T}}\Gamma v=\|\Gamma\|$;
but then replacing $u$ and $v$ with their entry-wise absolute
values, and replacing $\Gamma$ with its entry-wise absolute
value $\Gamma'$, we clearly get that $\|\Gamma'\|\ge\|\Gamma\|$.
However, $\max(\Gamma')=\max(\Gamma)$, so $\Gamma'$
remains feasible. This means we can always take the maximizing
matrix $\Gamma$ to be nonnegative, so $\Adv_1^\pm(f)=\Adv_1(f)$.
We can similarly assume that the unit vectors $u$ and $v$
maximizing $u^{\mathsf{T}}\Gamma v$ are nonnegative.

Finally, consider the maximizing matrix $\Gamma$ and the
maximizing unit vectors $u$ and $v$, all nonnegative,
and satisfying $\max(\Gamma)\le 1$. Note that
the expression $u^{\mathsf{T}}\Gamma v$ is nondecreasing in the entries
of $\Gamma$, since everything is nonnegative. Hence
to maximize $u^{\mathsf{T}}\Gamma v$, we can always take every nonzero
entry of $\Gamma$ to be $1$, since this maintains
$\max(\Gamma)\le 1$. In other words, the matrix maximizing
$\|\Gamma\|$ will always simply be $A_f$, and hence
$\Adv_1(f)$ is always exactly equal to $\lambda(f)$.

It remains to show that $\SA_1(f)=\SWA_1(f)=\MM_1(f)=\GSA_1(f)$.
The proof of this essentially
follows the arguments in \cite{SSpalekS06} for the regular
positive adversary, though some steps are a little simpler.
To start, we've seen that $\SA_1(f)=\lambda(f)$. Since
$A_f$ is symmetric, we have $\lambda(f)=v^{\mathsf{T}} A_f v$ for
some unit vector $v$, which we've established is nonnegative;
this vector is also an eigenvector, so $A_f v=\lambda(f)v$.
Consider the weight scheme $w(x,y)=v[x]v[y]A_f[x,y]$. Then
$wt(x)=\sum_y v[x]v[y]A_f[x,y]=v[x](A_f v)[x]=\lambda(f)v[x]^2$.
Hence if $w(x,x^i)>0$, we have
\begin{equation}
    \frac{\sqrt{wt(x)wt(x^i)}}{w(x,x^i)}=\frac{\lambda(f)v[x]v[x^i]}{v[x]v[x^i]A_f[x,x^i]}=\lambda(f).
\end{equation}
This means $\SWA_1(f)\ge\SA_1(f)$.
In the other direction, let $w$ be a feasible weight scheme
for $\SWA_1(f)$, let $\Gamma[x,y]=w(x,y)/\sqrt{wt(x)wt(y)}$,
and let $v[x]=\sqrt{wt(x)/W}$, where $W=\sum_x wt(x)$.
Then $\|v\|_2^2=\sum_x wt(x)/W=1$, and
\begin{equation}
    v^{\mathsf{T}}\Gamma v
=\sum_{x,y} \sqrt{wt(x)wt(y)}w(x,y)/W\sqrt{wt(x)wt(y)}
=(1/W)\sum_{x,y}w(x,y)=1.
\end{equation}
Hence $\|\Gamma\|\ge 1$. On the other hand, we have
$\max(\Gamma)=\max_{x,y} w(x,y)/\sqrt{wt(x)wt(y)}$.
This means that the ratio $\|\Gamma\|/\max(\Gamma)$
equals $\min_{x,y:w(x,y)>0}\sqrt{wt(x)wt(y)}/w(x,y)$,
which is $\SWA_1(f)$; thus $\SA_1(f)\ge\SWA_1(f)$.

Next we examine $\GSA_1(f)$. Consider a solution
$(Z,\Delta)$ to this semidefinite
program and define $\Gamma=Z\circ M\circ A_f$,
where $M$ is defined as $M=uu^{\mathsf{T}}$ and $u$ is defined by
$u[x]=1/\sqrt{\Delta[x,x]}$ when $\Delta[x,x]>0$
and $u[x]=0$ otherwise. Recall that $\Delta$ is diagonal
and that $\Delta-Z\circ D_i\succeq 0$ for all $i$.
Since positive semidefinite matrices are symmetric,
$Z\circ D_i$ must be symmetric for all $i$, so $Z$
is symmetric. Moreover, the diagonal of $Z\circ D_i$
is all zeros, so we must have $\Delta\ge 0$.
Further, if $\Delta[x,x]=0$ for some $x$, we must have
the corresponding row and column of $Z$ be all zeros.
If we let $\Delta'$ and $Z'$ be $\Delta$ and $Z$ with the
all-zero rows and columns deleted, then it is clear that
$\Delta-Z\circ D_i\succeq 0$ if and only if
$\Delta'-Z'\circ D_i\succeq 0$. Defining $M'$ as $M$
with those rows and columns deleted and $u'$ as $u$ with
those entries deleted, we have $M'=u'(u')^{\mathsf{T}}>0$.
Observe that $\Delta'-Z'\circ D_i\succeq 0$
if and only if $v^{\mathsf{T}}(\Delta'-Z'\circ D_i)v\ge 0$ for all
vectors $v$, which is if and only if
$(v\circ u')^{\mathsf{T}}(\Delta'-Z'\circ D_i)(v\circ u')\ge 0$
for all vectors $v$ (since we have $u'>0$). This, in turn,
is equivalent to $M'\circ (\Delta'-Z'\circ D_i)\succeq 0$.
Since $M'\circ \Delta'=I$, this is equivalent to
$I-M'\circ Z'\circ D_i\succeq 0$,
which is in turn equivalent to $I-M\circ Z\circ D_i\succeq 0$.
Since $Z\ge 0$ and we are maximizing $\langle Z,A_f\rangle$,
it never helps for $Z$ to have nonzero entries in places
where $A_f$ is $0$. Hence we can assume without loss of generality
that $Z=Z\circ A_f$, which means the constraint becomes
$I-\Gamma\circ D_i\succeq 0$, where we defined
$\Gamma=M\circ Z\circ A_f$. We thus have
$\|\Gamma\circ D_i\|\le 1$. On the other hand, letting
$v[x]=\sqrt{\Delta[x,x]}$, we have
\begin{equation}
    v^{\mathsf{T}}\Gamma v=\sum_{x,y}v[x]v[y]M[x,y]Z[x,y]A_f[x,y]
=\sum_{x,y:\Delta[x,x],\Delta[y,y]>0}Z[x,y]A_f[x,y]
=\langle Z,A_f\rangle.
\end{equation}
Hence $\SA_1(f)\ge\GSA_1(f)$.
The reduction in the other direction works similarly:
start with an adversary matrix $\Gamma$ with
$\max(\Gamma)\le 1$, and let $v$ be its
principle eigenvector. Then set $Z=\Gamma\circ (vv^{\mathsf{T}})$
and $\Delta=I\circ (vv^{\mathsf{T}})$. Then $I-\Gamma\circ D_i\succeq 0$,
which implies that $\Delta-Z\circ D_i\succeq 0$.
We also have $\tr\Delta=1$, $Z\ge 0$, and
$\langle Z,A_f\rangle=\|\Gamma\|$.

Finally, we handle $\MM_1(f)$. To do so, we first take the
dual of the semidefinite program for $\GSA_1(f)$.
This dual has the form
\begin{equation}
\begin{array}{lll}
\text{minimize}  &  \alpha &\\
\text{subject to}& \sum_i R_i\circ I\le\alpha I & \\
                & \sum_i R_i\circ D_i\ge A_f & \\
                & R_i\succeq 0 & \forall i\in[n]
\end{array}
\end{equation}
where the variables are $\alpha$ (a scalar) and matrices
$R_i$, each with rows and columns indexed by $\Dom(f)$.
Strong duality follows since when
$A_f$ is not all zeros, and the semidefinite program
in $\GSA_1(f)$ has a strictly feasible solution
(just take $Z$ to equal $\epsilon A_f$ for a small enough
positive constant $\epsilon$, and take $\Delta=I/|\Dom(f)|$).
This means the optimal solution of the minimization
problem above equals $\Adv_1(f)$. It remains
to show that this optimal solution $T$ also equals
$\MM_1(f)$.

Let $\alpha$ and $\{R_i\}_i$ be a feasible solution to
the semidefinite program above. Since $R_i\succeq 0$,
we have $R_i=X_iX_i^{\mathsf{T}}$ for some matrix $X_i$.
Define $w(x,i)=R_i[x,x]$.
Note that we also have $w(x,i)=\sum_a X_i[x,a]^2$.
Then by Cauchy--Schwarz,
$w(x,i)w(y,i)\ge
\left(\sum_a X_i[x,a]X_i[y,a]\right)^2
=(X_iX_i^{\mathsf{T}})[x,y]^2=R_i[x,y]^2$.
If $x$ and $y$ are such that $A_f[x,y]=1$, then they
disagree in only one bit $i$, and hence $D_i[x,y]=1$
for that $i$ and $D_j[x,y]=0$ for all $j\ne i$.
Since we have $\sum_i R_i\circ D_i\ge A_f$,
we conclude that for all such pairs $(x,y)$,
we have $w(x,i)w(y,i)\ge R_i[x,y]^2\ge A_f[x,y]^2=1$
on the bit $i$ where $x$ and $y$ differ; hence
the weight scheme $w$ is feasible.
Furthermore, for any $x$,
$\sum_i w(x,i)=\sum_i R_i[x,x]\le\alpha I[x,x]=\alpha$.
Hence $\MM_1(f)$ is at most the optimal value of this
semidefinite program.

In the other direction, consider a feasible weight scheme
$w$, and define $R_i[x,y]=\sqrt{w(x,i)w(y,i)}$.
Then $R_i=w(\cdot,i)w(\cdot,i)^{\mathsf{T}}$, where we treat
$w(\cdot,i)$ as a vector; hence $R_i\succeq 0$.
Moreover, $R_i\ge 0$, and for a pair $(x,y)$
with $A_f[x,y]=1$, there is some $i$ which is the unique
bit they disagree on, and hence $w(x,i)w(y,i)\ge 1$;
but this means that $R_i[x,y]\ge 1$, and so
$(R_i\cdot D_i)[x,y]\ge 1=A_f[x,y]$.
Finally, $\sum_i R_i[x,x]=\sum_i w(x,i)$,
which means that $\sum_i R_i\circ I\le \MM_1(f)\cdot I$, as desired.
\end{proof}

\subsection{Composition theorem}

Just like we have perfect composition theorems for degree (i.e., $\deg(f\circ g) = \deg(f)\deg(g)$) and deterministic query complexity (i.e., $\D(f \circ g) = \D(f)\D(g)$), we can show one for $\lambda(f)=\Adv_1(f)$.

\begin{theorem}\label{thm:composition}
For all (possibly partial) functions $f$ and $g$,
we have $\Adv_1(f\circ g)=\Adv_1(f)\Adv_1(g)$.
\end{theorem}

\begin{proof}
For the lower bound direction, note that $A_f$ is the matrix
with $A_f[x,y]=1$ if $f(x)\ne f(y)$ and the Hamming distance
between $x$ and $y$ is $1$ (with $A_f[x,y]=0$ otherwise).
$A_g$ is defined similarly. We wish to lower bound
$\|A_{f\circ g}\|$. To do so, we first introduce some notation.
Let $n$ be the input size of $f$ and let $m$ be the input size
of $g$. For $nm$-bit strings $x$ and $y$, we write
$x=x^{(1)}x^{(2)}\dots x^{(n)}$
and $y=y^{(1)}y^{(2)}\dots y^{(n)}$,
where $x^{(i)}$ and $y^{(i)}$ are $m$-bit strings.
We write $g(x)$ as shorthand for the string
$g(x^{(1)})g(x^{(2)})\dots g(x^{(n)})$,
and similarly for $g(y)$. For a string $x\in\{0,1\}^{nm}$,
we let $x^{(i,j)}$ denote the string $x$ with
the bit at position $(i,j)$ flipped; in particular,
the Hamming distance between $x$ and $x^{(i,j)}$ is $1$.
We will also use $\mathsf{s}_f(z)$ to denote the set of
sensitive bits of the string $z$ with respect to function $f$.

Let $v$ be the principal eigenvector of $A_f$ and let $u$
be the principal eigenvector of $A_g$. We can assume
they are nonnegative. Then $\|v\|_2=\|u\|_2=1$,
$A_fv=\lambda(f)v$, and $A_gu=\lambda(g)u$.
Let $v_0$ denote the component of $v$ on $0$-inputs
of $f$ and let $v_1$ be the component for $1$-inputs,
so that $v=[v_0,v_1]$. Define $u_0$ and $u_1$
similarly. Then since $A_f$ never has a $1$
in a position $(x,y)$ where $f(x)=f(y)$,
it decomposes into blocks of the form $[0, B_f; B_f^{\mathsf{T}}, 0]$, and
we must have $A_fv_0=\lambda(f)v_1$ and
$A_fv_1=\lambda(f)v_0$. Similarly,
$A_gu_0=\lambda(g)u_1$ and $A_gu_1=\lambda(g)u_0$.
We can assume without loss of generality that
$\|u_0\|_2^2=\|u_1\|_2^2=\|v_0\|_2^2=\|v_1\|_2^2=1/2$,
because otherwise, rebalancing the weights of $v_0$
and $v_1$ could increase $v^\mathsf{T} A_f v$ without
increasing $\|v\|$ (and similarly for $u_0$ and $u_1$).

Define the vector $\alpha$ with one entry for each input
to $f\circ g$ by
\begin{equation}\alpha[x]\coloneqq 2^{n/2} v[g(x)]u[x^{(1)}]u[x^{(2)}]\dots u[x^{(n)}].\end{equation}
Then
\begin{equation}\|\alpha\|_2^2=\sum_x\alpha[x]^2
=2^n\sum_{z\in\Dom(f)}
\sum_{y_1\in g^{-1}(z_1)}\dots\sum_{y_n\in g^{-1}(z_n)}
v[z]^2u[y_1]^2\dots u[y_n]^2\end{equation}
\begin{equation}=2^n\sum_{z\in\Dom(f)} v[z]^2\|u_{z_1}\|_2^2\dots\|u_{z_n}\|_2^2
=1.\end{equation}
We also have
\begin{align*}
\alpha^{\mathsf{T}}A_{f\circ g}\alpha
&=\sum_{x,x'}\alpha[x]\alpha[x']A_{f\circ g}[x,x']\\
&=\sum_{x}\sum_{i\in[n],j\in[m]}
    \alpha[x]\alpha[x^{(i,j)}]A_{f\circ g}[x,x^{(i,j)}]\\
&=\sum_{z\in\Dom(f)}
    \sum_{y_1\in g^{-1}(z_1)}\dots\sum_{y_n\in g^{-1}(z_n)}
    \sum_{i\in\s_f(z)}\sum_{j\in \s_g(y_i)}
    \alpha[y_1\dots y_n]\alpha[y_1\dots y_i^j\dots y_n]\\
&=\sum_{z\in\Dom(f)}
    \sum_{y_1\in g^{-1}(z_1)}\dots\sum_{y_n\in g^{-1}(z_n)}
    \sum_{i\in[n]}\sum_{j\in [m]}
    \alpha[y_1\dots y_n]\alpha[y_1\dots y_i^j\dots y_n]
    A_f[z,z^i]A_g[y_i,y_i^j]\\
&=\sum_{z\in\Dom(f)}\sum_{i\in[n]}v[z]v[z^i]A_f[z,z^i]\gamma(z,i),
\end{align*}
where
\begin{align*}
\gamma(z,i)&=\sum_{y_1\in g^{-1}(z_1)}\dots\sum_{y_n\in g^{-1}(z_n)}
    \sum_{j\in [m]}
    2^n u[y_1]\dots u[y_n]\cdot u[y_1]\dots u[y_i^j]\dots u[y_n]
    A_g[y_i,y_i^j]\\
&=\sum_{y_i\in g^{-1}(z_i)}\sum_{j\in [m]}
    2u[y_i]u[y_i^j]A_g[y_i,y_i^j]\prod_{k\ne i} 2\|u_{z_k}\|_2^2\\
&=2u_0^{\mathsf{T}}B_g u_1\\
&=\lambda(g).
\end{align*}
Hence
\begin{equation}\alpha^{\mathsf{T}}A_{f\circ g}\alpha
=\lambda(g)\sum_{z\in\Dom(f)}\sum_{i\in[n]}v[z]v[z^i]A_f[z,z^i]
=\lambda(g)v^{\mathsf{T}}A_fv=\lambda(g)\lambda(f).\end{equation}
This shows that $\lambda(f\circ g)\ge\lambda(f)\lambda(g)$.

For the upper bound direction, we use $\MM_1$.
Let $w_f$ be a feasible weight scheme for $f$, and let
$w_g$ be a feasible weight scheme for $g$. Define
weight scheme $w$ for $f\circ g$ by
$w(x,(i,j))=w_f(g(x),i)\cdot w_g(x^{(i)},j)$.
Then clearly $w$ is nonnegative and for each $x$,
\begin{equation}
\sum_{i,j} w(x,(i,j))=\sum_i w_f(g(x),i)\sum_j w_g(x^{(i)},j)\le wt_f(x)\max_i wt_g(x^{(i)}),    
\end{equation}
where we use $wt_f(z)$ to denote $\sum_i w_f(z,i)$ and similarly
for $wt_g(z)$. Hence the objective value of $w$ is at most
the product of the objective values of $w_f$ and $w_g$.
Finally, note that
\begin{equation}w(x,(i,j))w(x^{(i,j)},(i,j))
=w_f(g(x),i)w_f(g(x)^i,i)w_g(x^{(i)},j)w_g((x^{(i)})^j,j).\end{equation}
If $(i,j)$ is sensitive for $x$, then $i$ is sensitive for $g(x)$
and $j$ is sensitive for $x^{(i)}$, and hence we have
$w_f(g(x),i)w_f(g(x)^i,i)\ge 1$ and
$w_g(x^{(i)},j)w_g((x^{(i)})^j,j)\ge 1$, which means
$w(x,(i,j))w(x^{(i,j)},(i,j))\ge 1$. Therefore, $w$ is feasible,
and we have $\MM_1(f\circ g)\le\MM_1(f)\MM_1(g)$, as desired.
\end{proof}

\subsection{Upper bounds}

We now show a slightly better upper bound on $\lambda(f)$, that it is upper bounded by the geometric mean of the 0-sensitivity and 1-sensitivity, which can be a better upper bound than $\s(f)$. 

We provide two proofs of this. The first uses the $\lambda(f)$ formulation and uses a linear algebra argument about norms. This proof is due to Laplante, Lee, and Szegedy~\cite{LLS06}, who observed this about the measure $\K(f)$. 

To describe this proof, we briefly need to describe some matrix norms.
For a vector $v \in \Reals^n$, the $p$-norm for a positive integer $p$ is defined as $\norm{v}_p=(\sum_{i\in[n]}|v_i|^p)^{1/p}$. We also define $\norm{v}_\infty = \max_{i \in [n]} |v_i|$.
Note that $\norm{v}_1$ is simply the sum of the absolute values of all the entries of the vector.

Similarly, for a matrix $A \in \Reals^{n\times m}$, we define the induced $p$-norm of $A$ to be
\begin{equation}
    \norm{A}_p = \max\{\norm{Ax}_p: \norm{x}_p =1\}.
\end{equation}
The spectral norm $\norm{A}$ is the induced $2$-norm $\norm{A}_2$.
The 1-norm $\norm{A}_1$ is simply the maximum sum of absolute values of entries in any column of the matrix. The $\infty$-norm $\norm{A}_\infty$ is  the maximum sum of absolute values of entries in any row of the matrix. 

Lastly, we need a useful relationship between these norms sometimes called  H\"{o}lder's inequality for induced matrix norms (see \cite[Corollary 2.3.2]{GV13} for a proof): 
\begin{proposition}\label{prop:Holder}
For all matrices $A \in \Reals^{n\times m}$, we have $\norm{A} \leq \sqrt{\norm{A}_1 \norm{A}_\infty}$.
\end{proposition}

We can now prove the upper bound:

\begin{lemma}\label{lem:lambda_s}
For all (possibly partial) functions $f$, we have
$\lambda(f)\le\sqrt{\s_0(f)\s_1(f)}$.
\end{lemma}
\begin{proof}
We know that $\lambda(f) = \norm{A_f}$ and $A_f$ is a matrix of the form $\Bigl(\begin{smallmatrix}0& B\\ B^{\mathsf{T}}& 0 \end{smallmatrix}\Bigr)$ if we rearrange the rows and columns so that all $0$-inputs come first and are followed by $1$-inputs, since $A_f$ only connects inputs with different $f$-values. Thus we have
\begin{equation}
    \lambda(f) = \norm{A_f} = \norm{B} \leq \sqrt{\norm{B}_1\norm{B_\infty}} = \sqrt{\s_0(f)\s_1(f)},
\end{equation}
where we used H\"{o}lder's inequality (\Cref{prop:Holder}) and the fact that the maximum row and column sum of $B$ are precisely $\s_0(f)$ and $\s_1(f)$, respectively.
\end{proof}

Our second proof of this claim uses the $\MM_1(f)$ formulation which yields an arguably simpler proof.

\begin{lemma}
For all (possibly partial) functions $f$, we have
$\Adv_1(f)\le\sqrt{\s_0(f)\s_1(f)}$.
\end{lemma}

\begin{proof}
Using the $\MM_1(f)$ version of $\Adv_1(f)$,
set $w(x,i)=\sqrt{\s_0(f)}/\sqrt{\s_1(f)}$ if $f(x)=1$,
and set $w(x,i)=\sqrt{\s_1(f)}/\sqrt{\s_0(f)}$ if $f(x)=0$.
Then if $x$ and $y$ differ in a single bit $i$,
we clearly have $w(x,i)w(y,i)=1$. On the other hand,
$\sum_i w(x,i)\le \s_1(f)\cdot \sqrt{\s_0(f)}/\sqrt{\s_1(f)}
=\sqrt{\s_0(f)\s_1(f)}$
for $1$-inputs $x$, and analogously
$\sum_i w(y,i)\le \sqrt{\s_0(f)\s_1(f)}$
for $0$-inputs $y$.
\end{proof}

Using this better bound on $\lambda(f)$ and Huang's result, we also get that for all total Boolean functions $f$, 
\begin{equation}
    \deg(f) \leq \s_0(f) \s_1(f).
\end{equation}
This result was also recently observed by Laplante, Naserasr, and Sunny~\cite{LNS20}.  Unlike their proof, the following uses Huang's theorem in a completely black-box way.

\begin{proposition}
Assume that $\deg(f) \leq \s(f)^2$ for all total Boolean functions $f$. Then we also have $\deg(f) \leq \s_0(f) \s_1(f)$.
\end{proposition}
\begin{proof}
Let $\s_0(f) = k$ and $\s_1(f) = \ell$. We know that $\deg(f) \leq \max\{k,\ell\}$ by assumption. 
Let $\AND_{k} \circ \OR_{\ell}$ be the AND function on $k$ bits composed with the OR function on $\ell$ bits. Clearly $\s_0(\AND_{k} \circ \OR_{\ell})=\ell$ and $\s_1(\AND_{k} \circ \OR_{\ell}) = k$. 
Furthermore, because the function is monotone, the sensitive bits for a $0$-input are bits set to $0$, and the sensitive bits for a $1$-input are bits set to $1$. 
This means that composing this function with $f$ with yield a function where the one-sided sensitivity will be upper bounded by the product of one-sided sensitivity of the individual functions. Hence for all $b \in \B$, we have
\begin{equation}
    \s_b(\AND_{k} \circ \OR_{\ell} \circ f) \leq \s_b(\AND_{k} \circ \OR_{\ell}) \s_b(f) \leq k\ell.
\end{equation}
Using the assumption on the function $\AND_{k} \circ \OR_{\ell} \circ f$, we get 
\begin{equation}
    \deg(\AND_{k} \circ \OR_{\ell} \circ f) \leq (\s(\AND_{k} \circ \OR_{\ell} \circ f))^2 \leq (k\ell)^2.
\end{equation}
Finally, it is well known that $\deg(f \circ g) = \deg(f)\deg(g)$~(see, e.g., \cite{Tal13}), and hence $\deg(\AND_{k} \circ \OR_{\ell} \circ f) = k\ell \deg(f)$, which implies $\deg(f) \leq k\ell$.
\end{proof}

\subsection{Lower bounds}

Finally, we describe some lower bounds on $\lambda(f)$. These follow from known results, but we reproduce them here for completeness.

\begin{lemma}\label{lem:lambdalower}
For all (possibly partial) functions $f$, $\s(f)\leq \lambda(f)^2$.
\end{lemma}
\begin{proof}
Consider any input $x$ with sensitivity $\s(f)$. This means $x$ has $\s(f)$ neighbors on the hypercube with different $f$ value. The sensitivity graph restricted to these $\s(f)+1$ inputs is a star graph centered at $x$. The spectral norm of the adjacency matrix of the star graph on $k+1$ vertices is $\sqrt{k}$. Since the spectral norm of $A_f$ is lower bounded by that of a submatrix, we have $\lambda(f) \geq \sqrt{\s(f)}$.
\end{proof}

This relationship is tight for the $\OR_n$ function which has $\s(\OR_n)=n$ and $\lambda(\OR_n)=\sqrt{n}$. Although $\OR_n$ has unbalanced sensitivities, with $\s_0(\OR_n)=n$ and $\s_1(\OR_n)=1$, there are functions $f$ with $\s(f)=\s_0(f)=\s_1(f)=n$ and $\lambda(f)=\sqrt{n}$. One example of such a function is $x_1 \oplus \OR(x_2,\ldots,x_n)$. Another example of such a function with a quadratic gap between $\s(f)$ and $\lambda(f)$ is the function that is $1$ if and only if the input string has Hamming weight $1$. This function has $\s_0(f)=n$ since the all zeros string is fully sensitive and $\s_1(f)=n$ since every Hamming weight $1$ string is also fully sensitive. But we know that this problem can be solved by Grover's algorithm with $O(\sqrt{n})$ queries, and hence $\lambda(f) = O(\Q(f)) = O(\sqrt{n})$.

We can also lower bound $\norm{A_f}$ by $\norm{A_f} \geq  |v^{\mathsf{T}} A_f v|$ for any vector $v$ with $\norm{v}=1$. If we take $v$ to be the normalized all ones vector, this is just the average sensitivity.  
\begin{lemma}\label{lem:averagesen}
For all (possibly partial) functions $f$, $\lambda(f) \geq \E_{x}[\s_x(f)]$.
\end{lemma}

For example, this shows that $\lambda(\Parity_n)=n$. The bound in \Cref{lem:averagesen} can be improved by only taking the expectation on the right over a subset of the inputs of $f$, which then equals another complexity measure originally defined by Khrapchenko~\cite{Khr71}. 
See \cite{Kou93} for more on this relationship and Khrapchenko's bound.
\end{document}